\documentclass[aps,prl,superscriptaddress,twocolumn,floatfix,10pt]{revtex4-2}

\usepackage{amsmath,mathtools,amsthm,amssymb}

\usepackage{tabularx}
\usepackage{units}
\usepackage{graphicx}
\usepackage{color}
\usepackage{xcolor}
\usepackage{bbold}
\usepackage{titlesec}
\usepackage{physics}
\usepackage{enumitem}
\usepackage{pifont}
\usepackage{times}
\usepackage{comment}
\usepackage{array}
\newcolumntype{H}{>{\setbox0=\hbox\bgroup}c<{\egroup}@{}}
\usepackage{graphics}

\usepackage[normalem]{ulem}

\definecolor{myurlcolor}{rgb}{0,0,0.7}
\definecolor{myrefcolor}{rgb}{0.1,0,0.9}

\usepackage[
    breaklinks,
    pdftex,
    colorlinks=true,
    linkcolor=myrefcolor,
    citecolor=myurlcolor,
    urlcolor=myurlcolor
]{hyperref}
\usepackage[capitalize]{cleveref}
\crefname{enumi}{constraint}{constraints}
\Crefname{enumi}{Constraint}{Constraints}
\crefname{case}{case}{cases}
\Crefname{case}{Case}{Cases}

\usepackage[linesnumbered, ruled,vlined]{algorithm2e}

\newcommand{\smref}[1]{SM~\ref{#1}}

\newtheorem{theorem}{Theorem}
\newtheorem{lemma}{Lemma}

\newtheorem{definition}{Definition}

\newtheorem{remark}{Remark}

\newcommand{\haar}[0]{\operatorname{Haar}}

\newcommand{\pur}{\operatorname{Pur}}

\newcommand{\st}[1]{\ketbra{#1}{#1}}
\newcommand{\poly}{\operatorname{poly}}

\renewcommand{\Pr}{{\mathsf{Pr}}}
\newcommand{\TV}{{\mathsf{TV}}}
\newcommand{\E}{{\mathbb{E}}}
\newcommand{\negl}{{\mathsf{negl}}}

\allowdisplaybreaks

\begin{document}

\title{Pseudomagic quantum states}

\author{Andi Gu}
\email{andigu@g.harvard.edu}
\affiliation{Department of Physics, Harvard University, 17 Oxford Street, Cambridge, MA 02138, USA}

\author{Lorenzo Leone}
\email{lorenzo.leone@fu-berlin.de}
\affiliation{Department of Physics, University of Massachusetts Boston,
    100 Morrissey Boulevard, Boston, MA 02125, USA}
\affiliation{Dahlem Center for Complex Quantum Systems, Freie Universit\"at Berlin, 14195 Berlin, Germany}

\author{Soumik Ghosh}
\affiliation{Department of Computer Science, University of Chicago, Chicago, Illinois 60637, USA}

\author{Jens Eisert}
\affiliation{Dahlem Center for Complex Quantum Systems, Freie Universit\"at Berlin, 14195 Berlin, Germany}
\affiliation{Helmholtz-Zentrum Berlin f\"ur Materialien und Energie, 14109 Berlin, Germany}
\affiliation{Fraunhofer Heinrich Hertz Institute, 10587 Berlin, Germany}

\author{Susanne F. Yelin}
\affiliation{Department of Physics, Harvard University, 17 Oxford Street, Cambridge, MA 02138, USA}

\author{Yihui Quek}
\affiliation{Department of Physics, Harvard University, 17 Oxford Street, Cambridge, MA 02138, USA}
\affiliation{Department of Computer Science, Harvard John A. Paulson School Of Engineering And Applied Sciences, 150 Western Ave, Boston, MA 02134}

\begin{abstract}
Notions of nonstabilizerness, or ``magic'', quantify how non-classical quantum states are in a precise sense: states exhibiting low nonstabilizerness preclude quantum advantage. 
We introduce `pseudomagic' ensembles of quantum states that, despite low nonstabilizerness, are computationally indistinguishable from those with high nonstabilizerness. Previously, such computational indistinguishability has been studied with respect to entanglement, introducing the concept of pseudoentanglement. However, we demonstrate that pseudomagic neither follows from pseudoentanglement nor implies it. In terms of applications, the study of pseudomagic offers fresh insights into the theory of quantum scrambling: it uncovers states that, even though they originate from non-scrambling unitaries, remain indistinguishable from scrambled states to any physical observer. Additional applications include new lower bounds on state synthesis problems, property testing protocols, and implications for quantum cryptography. Our work is driven by the observation that only quantities measurable by a \textit{computationally bounded observer} -- intrinsically limited by finite-time computational constraints -- hold physical significance. Ultimately, our findings suggest that nonstabilizerness is a `hide-able' characteristic of quantum states: some states are much more magical than is apparent to a computationally bounded observer.

\end{abstract}

\maketitle

The boundary between quantum and classical computation is a central question in current research, with a focus on identifying uniquely quantum resources that contribute to a quantum advantage. One such resource is nonstabilizerness (``magic''), which is a measure of the non-Clifford resources needed to prepare a quantum state~\cite{Bravyi_2005, Howard2017, Bartlett2014}. Nonstabilizerness is directly connected to the hardness of classically simulating a quantum state~\cite{gottesman1997stabilizer, stabilizer2004, Bravyi_2016,bravyi_trading_2016,bravyi_simulation_2019,seddon_quantifying_2021,Wigner,veitch_resource_2014,Emerson}, the yield of magic state distillation protocols~\cite{howard_application_2017,OGorman2017, Bravyi_2005, Bravyi2012,campbell_unified_2017,campbell_unifying_2017,campbell_catalysis_2011, dawkins_qutrit_2015, campbell_bound_2010,anwar_qutrit_2012,krishna_low_2019}, the overhead required for fault-tolerant quantum computation~\cite{kitaev_faulttolerant_2003,Campbell2010, campbell2017,xu2023constantoverhead}, and the degree of chaos in a system~\cite{leone_stabilizer_2022,leone_nonstabilizerness_2023,goto_probing_2022,garcia_resource_2023}. Given these connections, one might expect that quantum states with high nonstabilizerness are inherently different, and more non-classical, than states with low 
values.

In this work, we challenge this intuition by constructing ensembles of states that are poor in magic resources but, nonetheless, are \emph{computationally indistinguishable} from an ensemble of states that are rich in magic resources. Because they masquerade as highly magical ensembles, we call the former ``pseudomagic'' ensembles. Moreover, their nonstabilizerness can also be \emph{tuned}: for any value of nonstabilizerness strictly greater than $\log(n)$ and up to $n$, there is a pseudomagic ensemble with that amount of nonstabilizerness. While the `pseudoentangled' ensembles introduced in 
Ref.~\cite{aaronson_quantum_2023} happen to also display the above `pseudomagic' properties, we prove that the amount of each resource in them (i.e., entanglement and magic) can be \emph{tuned independently}: the existence of pseudoentangled ensembles does not imply their pseudomagic counterparts, nor vice versa.
While we quantify nonstabilizerness in the rest of this work with the measure of stabilizer R\'enyi entropy~\cite{leone_stabilizer_2022}, we explain how to generalize our construction to many other popular magic measures, such as robustness of magic~\cite{howard_application_2017}, stabilizer fidelity
and extent~\cite{bravyi_simulation_2019}, max relative entropy of magic~\cite{liu_manybody_2022}. %

As physically motivated applications, we discuss the implication of pseudomagic for quantum chaos theory. Our results imply, counterintuitively, that some states generated by non-scrambling unitaries are computationally indistinguishable from states generated by scrambling unitaries. Furthermore, we show that the existence of pseudomagic states immediately implies the existence of a quantum cryptographic primitive known as EFI pairs~\cite{brakerski2022computational}. Finally, we employ our findings to obtain lower bounds for black-box magic state distillation 
and property testing protocols.

\emph{The computationally bounded observer. } 
The central claim of this work is that an idea from computer science -- namely, computational indistinguishability -- has significant ramifications for the ways in which we understand physics. Say we have two $n$-particle systems $S_1$ and $S_2$ that differ in some physical characteristic $C$, and we aim to distinguish between the two systems. What if the time needed for \emph{any} distinguishing method is on par with the age of the universe? If so, we then have two systems which purportedly differ in some physical attribute $C$, yet nevertheless can never be distinguished in any reasonable amount of time -- so in what sense can we say $C$ is a genuine physical attribute? In the language of computer science, a  distinguishing algorithm is efficient if its requisite \textit{computational time} (i.e., number of elementary operations) scales polynomially with the number of particles, denoted by $\poly(n)$, and inefficient if it scales exponentially as $\exp(n)$. This (fuzzy) distinction delineates between a feasible distinguishing scheme and an impractical one. Indeed, if every possible distinguishing algorithm scales exponentially with the number of particles $n$, we say the two systems $S_1$ and $S_2$ are `computationally indistinguishable', as for even a modest system size $n\sim 120$, an exponential distinguisher's runtime will surpass the age of the universe $\sim 10^{18}\,$s, %
impossible to execute in practice. This idea motivates us to introduce the notion of  the \textit{Computationally Bounded Observer} (CBO), an observer constrained to measurement schemes that operate within polynomial time. The results presented in this work demonstrate that the (many-body) physics observed by a CBO differs profoundly compared to an unconstrained observer. This introduces a unique perspective into quantum many-body systems, transforming the laboratory from a mere verifier of quantum theories into an integral part of the theory itself, where the limitations of observers assume a central role.

\emph{Pseudomagic. } 
First, we review some useful definitions associated with nonstabilizerness. Let $\mathbb{P}_n$ be the Pauli group on $n$ qubits, $\mathcal{C}_n$ be the Clifford group, and let $\Sigma$ be the set of pure stabilizer states. A stabilizer operation $\mathcal{S}$ is a quantum channel obeying $\mathcal{S}(\Sigma)=\Sigma$, which is to say that $\mathcal{S}$ preserves the set of stabilizer states~\cite{veitch_resource_2014}. 
There are many ways to quantify nonstabilizerness~\cite{howard_application_2017,bravyi_simulation_2019,leone_stabilizer_2022,liu_manybody_2022,tirrito_quantifying_2023,haug2023stabilizer,haug_efficient2023,haug_scalable_2023,ME,bu2023stabilizer,chamon_quantum_2022,turkeshi2023measuring,saxena_quantifying_2022,koukoulekidis_constraints_2022}, but we limit our attention to the \textit{magic measures} introduced in \cref{table:magicmonotone}. The second key concept in this work is the notion of ``computational indistinguishability" for two ensembles of states, which means that no CBO (i.e., polynomially-bounded algorithm) can tell the difference between the two ensembles. We refer the reader to Ref.\ \cite{jirandom2018} for a detailed treatment of this concept and its relevance for quantum cryptography. Having established the two notions of magic and computational indistinguishability, we now introduce pseudomagic.

\begin{table}
    \setlength{\tabcolsep}{6pt} %
    \renewcommand{\arraystretch}{1.4} %
    \begin{centering}
        \resizebox{\columnwidth}{!}{%
            \begin{tabular}{|c|c|}
                \hline
                {\textbf{Magic measure}}  & {\textbf{Definition}} \\
                \hline
                \hline
                Stabilizer entropy~\cite{leone_stabilizer_2022}   &

                $M_{\alpha}({\psi})=\frac{1}{1-\alpha}\log \frac{1}{d} \sum_P \tr^{2\alpha}(P\psi)$                             \\
                \hline
                Robustness of magic~\cite{howard_application_2017} & $\mathcal{R}(\psi)=\log\min\{\norm{c}_1 \mid  \psi=\sum_{i}c_i\st{\sigma_i}\}$                                      \\
                \hline
                Stabilizer fidelity~\cite{bravyi_simulation_2019}  & $\mathcal{F}_{\textrm{stab}}(\psi)=-\log \max_{\sigma \in \Sigma}\abs{\braket{\psi}{\sigma}}^2$       \\
                \hline
                Stabilizer extent~\cite{bravyi_simulation_2019}    & $\xi(\psi)=\log\min(\sum_{\phi}\abs{c_{\phi}})^2$                                                                         \\
                \hline
                Max-relative entropy~\cite{liu_manybody_2022}      & $D_{\max}(\psi)=\log\min\{\lambda \mid \lambda\sigma-\psi\ge 0\}$                                                    \\
                \hline
            \end{tabular}
        }
        \caption{\label{table:magicmonotone} Magic measures and their definitions.}
        \label{tablemagicmonotone}
    \end{centering}
\end{table}

\begin{figure}
    \centering
    \includegraphics[width=\columnwidth]{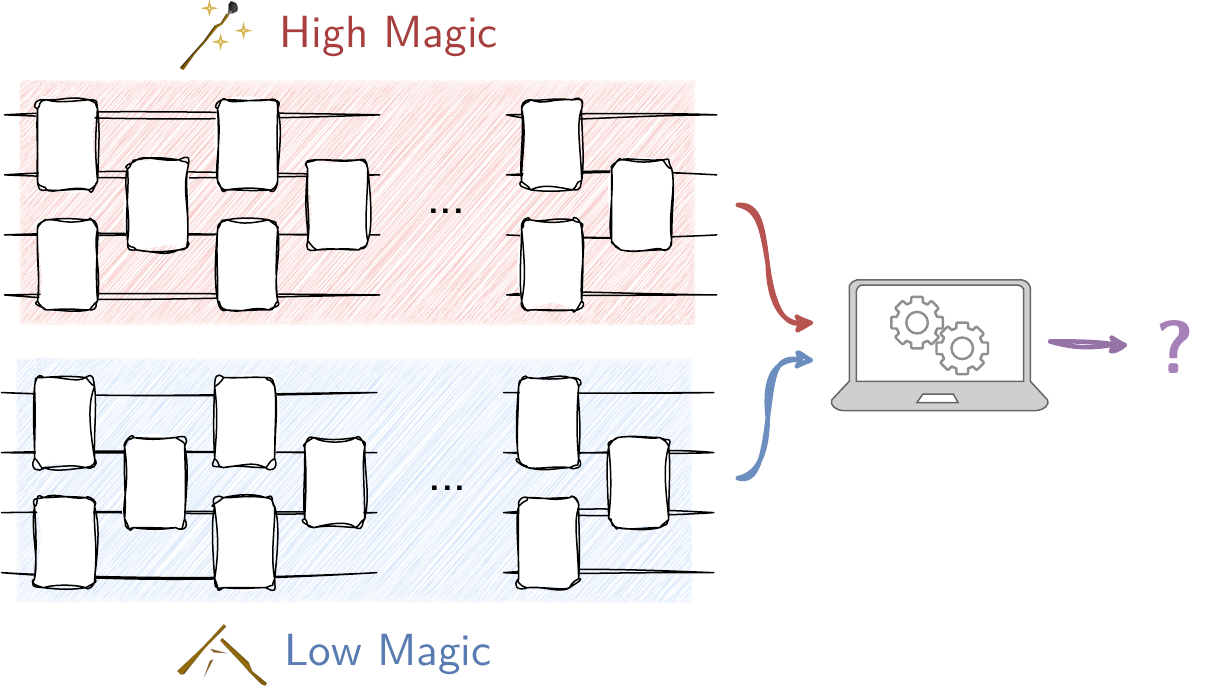}
    \caption{The `pseudomagic' state ensembles discussed in this work are computationally indistinguishable from highly magical states produced by scrambling dynamics (e.g., Haar random states).}
    \label{fig:chaos}
\end{figure}

\begin{definition}[Pseudomagic]\label{def:pseudomagic}
Let $\mathcal{M}$ be a magic measure. A pseudomagic pair with gap $f(n)$ vs. $g(n)$ (where $f(n)>g(n)$) consists of two state ensembles:
\begin{enumerate}[label=(\alph*)]
    \item a `high magic' ensemble of $n$-qubit quantum states $\{|\psi_{k_1}\rangle\}$ such that $\mathcal{M}({\psi_{k_1}})=
    f(n)$ with high probability over $k_1$, and
    \item a `low magic' ensemble of $n$-qubit quantum states $\{|\phi_{k_2}\rangle\}$ such that $\mathcal{M}({\phi_{k_2}})=
    g(n)$ with high probability over $k_2$,
\end{enumerate}
such that the two ensembles are computationally indistinguishable, even when given polynomially many copies. %
\end{definition}
Qualitatively, states from the ensemble $\{\ket{\phi_{k_2}}\}$ mimic much more `magical' states to all CBOs, even though they themselves are `low magic'. In the rest of this work we use stabilizer R\'enyi entropy as our magic measure, given by
\begin{equation}
    M_{\alpha}({\psi})=\frac{1}{1-\alpha}\log \frac{1}{2^n}\sum_{P \in \mathbb{P}_n}\tr^{2\alpha}(P\psi), \label{eq:malpha}
\end{equation}
which is a magic monotone for $\alpha\ge 2$~\cite{leone2024stabilizermonotone}; see~\smref{sec:prelim} for a detailed discussion~\footnote{Moreover, the stabilizer entropy possesses two desirable traits: it is experimentally measurable, and serves as a lower bound for all other magic monotones listed in \cref{table:magicmonotone}. These feautures enable us to demonstrate the existence of states with a significant and precise pseudomagic gap of $\log n$ versus $n$. See Supplemental Material for an in-depth discussion about the role of $M_{\alpha}$ in defining pseudomagic.}. Combining the `computational indistinguishability' property with the properties of stabilizer R\'enyi entropy, we show that $g(n)$, the nonstabilizerness of the low-magic ensemble, can be no smaller than $\omega(\log n)$ \footnote{To remind the reader, $f=o(g(n)) \implies \lim {f(n)}/{g(n)}=0$, $f=\omega(g(n)) \implies \lim {f(n)}/{g(n)}=\infty$, and finally $f=\Theta(g(n)) \implies \exists C_1,C_2 > 0$ such that $C_1 \leq \lim {f(n)}/{g(n)} \leq C_2$.}:

\begin{lemma}[Bound to stabilizer entropies]\label{lem:stab-bound}
Let $\mathcal{E}$ be the low magic ensemble of a pseudomagic pair. Then the $\alpha$-R\'enyi stabilizer entropies obey $M_{\alpha}({\psi})=\omega(\log n)$, with high probability over the choice of $\ket{\psi} \in \mathcal{E}$.
\end{lemma}
In the next section, we turn to the construction of pseudomagic pair with maximal gap, i.e. $\omega(\log n)$ vs. $O(n)$. %

\emph{Construction of pseudomagic ensembles.} For any function $f:\{0,1\}^n\to \{0,1\}$ and subset $S\subseteq \{0,1\}^n$, we define the associated subset phase state as~\cite{aaronson_quantum_2023}
\begin{equation}
    \ket{\psi_{f,S}}=\frac{1}{\sqrt{|S|}}\sum_{x\in S}(-1)^{f(x)}\ket{x}.
    \label{eq:subsetphasestates}
\end{equation}
In Ref.~\cite{aaronson_quantum_2023}, it was shown the ensemble $\mathcal{E}=\{\ket{\psi_{f,S}}\}$ of subset phase states (for pseudorandom functions $f$ and subsets $S$) is both efficiently preparable, and computationally indistinguishable from Haar random states. In this section, we show that $\mathcal{E}$ also saturates the nonstabilizerness lower bound in \cref{lem:stab-bound}. Since the ensemble of Haar-random states $\mathcal{E}_{\haar}$ have stabilizer entropy $\Theta(n)$ with overwhelming probability~\smref{subsec:lem1proof}, $(\mathcal{E},\mathcal{E}_{\haar})$ form a pseudomagic pair with gap $\omega(\log n)$ vs. $O(n)$.  We then show that applying a polynomial depth quantum circuit $U$ to $\mathcal{E}$ results in a high-magic ensemble $\mathcal{E}_U$ that, thanks to the transitivity of computational indistinguishability, is computationally indistinguishable from the low-magic ensemble $\mathcal{E}$, thus forming another pseudomagic pair $(\mathcal{E},\mathcal{E}_U)$ (see \cref{fig:pseudotransitions}).

\begin{theorem}[Subset phase states display pseudomagic\label{thm:pseudomagic ensembles}]
    For any $ k\in[\omega(\log n), n]$, there exists an ensemble
    $\mathcal{E}=\{\ket{\psi_{f,S}}\}$ for $|S|=2^k$ that has $M_{\alpha}({\psi_{f,S}})= O(k)$. For $\alpha \leq 2$, this bound is tight: $M_{\alpha}({\psi_{f,S}})= \Theta(k)$. Furthermore, 
    \begin{enumerate}[label=(\alph*)]
        \item $\mathcal{E}$ is computationally indistinguishable the ensemble of Haar random states $\mathcal{E}_{\haar}$. Therefore, $(\mathcal{E},\mathcal{E}_{\haar})$ forms a pseudomagic pair with gap $\omega(\log n)$ vs. $O(n)$.
        \item There exists a quantum circuit $U$ composed solely of single-qubit gates (i.e., it has depth 1) such that $(\mathcal{E}$, $\mathcal{E}_U)$, where $\mathcal{E}_U=\{U\ket{\psi_{f,S}}\}$, also forms a pseudomagic pair with maximal gap $\omega(\log n)$ vs. $O(n)$.
    \end{enumerate}
\end{theorem}

It is natural to ask if, besides the stabilizer R\'enyi entropy, other magic measures (see \cref{tablemagicmonotone}) can be used to define pseudomagic. In \smref{subsec:lemmarobustness}, we show that both the $\alpha$-stabilizer entropies $M_{\alpha}(\psi)$ and the log-robustness of magic $\mathcal{R}(\psi)$ are measures of pseudomagic with gap $\Theta(n)$
vs.~$\Theta(\poly\log n)$. As a corollary, we have a sufficient condition for a magic measure to be a good measure for defining pseudomagic. For any magic monotone $\mathcal{M}$, if there exists an $\alpha$ such that $\Omega(M_\alpha(\psi))\leq \mathcal{M}(\psi) \leq O(\mathcal{R}(\psi))$ for pure states $\psi$, then $\mathcal{M}$ is also a measure of pseudomagic with gap $\Theta(n)$
vs.~$\Theta(\poly\log n)$. In \smref{subsec:sufficient}, we use this fact to show that the magic monotones $\mathcal{M}=\mathcal{F}_{\textrm{stab}}(\psi), \xi(\psi),D_{\max}(\psi)$, as defined in \cref{table:magicmonotone}, are also measures of pseudomagic with gap $\Theta(n)$
vs.~$\Theta(\poly\log n)$. Notably, these measures serve as genuine magic monotones even in the context of mixed-state magic resource theory.

\emph{Implications to quantum scrambling.} Having defined the notion of pseudomagic, we now use it to address some important aspects of quantum information scrambling, and explore how CBOs challenge its foundational principles. One way of defining a scrambling unitary evolution $U$ is to say that it is scrambling if it attains the Haar value for its $2k$-point
\emph{out-of-time-order correlators} (OTOCs)~\cite{roberts_chaos_2017,hosur_chaos_2016}. These OTOCs are denoted by $C_{2k}$ and are defined as
\begin{equation}
    C_{2k}(U)\coloneqq \frac{1}{d}\tr(\tilde{P}_1 Q_1\tilde{P}_2 Q_2\cdots\tilde{P}_k Q_k),
\end{equation}
where $P_i$ and $Q_i$ are non-identity Pauli operators for $i=1,\ldots, k$ and $\tilde{P}_i \coloneqq  U^{\dag}P_i U$. To be precise, $U$ is scrambling if $C_{2k}(U)=\tilde{O}(C_{2k}(U_{\haar}))$, where $\tilde{O}$ stands for an irrelevant polynomial overhead. Indeed, typically  $C_{2k}(U_{\haar})=O(\exp(-\gamma n))$ with $\gamma$ depending on the particular choice of the correlator $C_{2k}$~\cite{roberts_chaos_2017}. 

Magic resource theory is linked to quantum information scrambling~\cite{leone_isospectral_2021,oliviero_random_2021,PhysRevD.106.126009}. Clifford unitaries can be scramblers, but they scramble information in a relatively simple way~\cite{PhysRevLett.132.080402}. Any unitary exhibiting complex information scrambling for its OTOCs must contain $\Omega(n)$ non-Clifford gates \cite{leone_quantum_2021,oliviero_transitions_2021}. Information scrambled by a unitary evolution with less than $n$ non-Clifford gates can be \textit{unscrambled} and reconstructed~\cite{PhysRevLett.132.080402,PhysRevA.109.022429}. %
Consequently, the mere existence of pseudomagic states that are also pseudorandom (i.e., subset phase states), suggests the existence of non-complex scramblers that nonetheless generate states indistinguishable from states generated by maximal scramblers. More precisely, we establish that such states must be generated by a unitary evolution that exhibits exponentially separated OTOCs from the typical Haar value.
\begin{theorem}[Hidden quantum scrambling]\label{chaostheorem}
Let $\mathcal{E}$ be an ensemble of pseudomagic states that is also pseudorandom. Let $\ket{\psi}\in\mathcal{E}$ and let $U$ such that $\ket{\psi}=U\ket{0}^{\otimes n}$. The $2k$-point OTOCs of $U$ (for $k\geq 4$) are exponentially separated from the Haar value,
\begin{equation}
    C_{2k}(U)=\Omega(\exp(n)C_{2k}(U_{\haar})).
\end{equation}
Therefore, although it generates a state that is on-average computationally indistinguishable from Haar-random, $U$ is not fully scrambling.
\end{theorem}

The proof of this can be found in \smref{sec:chaosproof}. %
The curious implication of \cref{chaostheorem} is this: any physical observer, inherently constrained by computational limits (i.e., a CBO), cannot differentiate between a maximally scrambling evolution and non-scrambling one solely based on the observed resultant state (see \cref{fig:chaos}). As a result, a CBO interprets quantum scrambling evolutions differently from an unrestricted observer, exposing a significant conceptual challenge in characterizing quantum information scrambling. Additionally, our findings could potentially also lead to implications for the theory of quantum chaos. Indeed, by defining quantum chaotic evolution as maximally scrambling unitary operators, then \cref{chaostheorem} would imply the impossibility for a CBO to distinguish between a chaotic quantum evolution from a non-chaotic one. However, while some authors~\cite{roberts_chaos_2017} are inclined to think about quantum chaotic evolution precisely as maximal scramblers~\cite{roberts_chaos_2017}, this concept has been challenged in Ref.~\cite{PhysRevLett.131.180403}, where it has been proven that scrambling (when restricted to $k=2$) does not imply chaos. The question of whether maximal scrambling, as probed by higher order OTOCS, implies chaos remains an exciting venue for future research.

\emph{Implications to quantum cryptography.} An essential primitive for classical cryptography is the concept of a \emph{one-way function} (OWF), which is a function that is efficient to evaluate but hard to invert. However, the story is much different in the \emph{quantum} world: OWFs are unnecessary for some quantum cryptographic constructions to hold~\cite{kretschmer2021pseudorandomness,morimae2022commitments}. This leads naturally to a question of whether there is an indespensible primitive for quantum cryptography that serves a similar role to OWFs for classical cryptography. In Ref.~\cite{brakerski2022computational}, the authors introduce EFI pairs as this quantum analogue and show that it is necessary for many secure quantum cryptographic schemes, including bit commitment~\cite{lin2014,yan2020}, oblivious transfer~\cite{bartusek2021one,grilo2021oblivious}, multiparty quantum computation~\cite{ananth2022cryptography}, and zero knowledge proofs~\cite{ananth2021}. EFI pairs are state ensembles generated by efficient circuits that are statistically far but computationally indistinguishable. In light of the proposed significance of EFIs, we show the following.

\begin{theorem}[Cryptographic implications]
    \label{thm:crypto}
    Consider an ensemble of efficiently preparable pseudomagic states that have stabilizer entropy $M_{1}=\Theta(g(n))$ with high probability, where $g(n)$ is tunable in the range $\omega(\log n)$ and $O(n)$. Then, the pseudomagic ensemble, along with the high nonstabilizerness ensemble, forms an EFI pair.
\end{theorem}
For a proof, see \smref{sec:crypto}. Crucially, \cref{thm:crypto} holds even in a world without quantum-secure OWFs: it says that the bare existence of pseudomagic states with tunable stabilizer entropy implies the existence of EFI pairs and the world of cryptographic applications they unlock. This strengthens the case that EFI pairs are a more fundamental primitive for quantum cryptography than OWFs.

\emph{No efficient black-box magic-state distillation.} Several architectures for universal fault-tolerant quantum computing rely on applying stabilizer operations to carefully prepare resource states called magic states~\cite{bravyi_universal_2005,campbell2017, knill_quantum_2005, bravyi_magicstate_2012}. An example is the canonical magic state vector $\ket{T}=\ket{0}+e^{i\pi/4}\ket{1}$, which when provided as an input to auxiliary qubits, enables $T$-gate implementation using only stabilizer operations. However, not all nonstabilizer states are useful for implementing non-Clifford gates~\cite{beverland_lower_2020}, especially noisy ones. This motivates the question, can we develop efficient (i.e., polynomial-sized circuit description)  stabilizer protocols that can transform generic nonstabilizer states $\rho$ into useful nonstabilizer states, such as $\ket{T}$? In line with the spirit of analogous tasks for entanglement resource theory~\cite{harrow2005applications,hayashi2002universal}, we term this task \emph{black-box magic-state distillation}.

\begin{theorem}[Black-box magic state distillation]\label{thm:distill}
Given a magic monotone $\mathcal{M}$ such that $\Omega(M_{\alpha}(\psi))\le\mathcal{M}(\psi)\le O(\mathcal{R}(\psi))$ for all pure states $\psi$, any efficient stabilizer protocol that synthesizes a state vector $\ket{B}\bra{B}$ from an arbitrary (potentially mixed) input state $\rho$ requires
\begin{equation}
    \Omega\left(\frac{\mathcal{M}(\ket{B}\bra{B})}{\log^{1+c}\mathcal{M}(\rho)} \right)
\end{equation}
copies of $\rho$, for any constant $c>0$. Remarkably the above is valid for $\mathcal{M}=\mathcal{F}_{\mathrm{stab}},\xi, D_{\max}$ in \cref{table:magicmonotone}.
\end{theorem}

The proof of this can be found in \smref{sec:distillproof}. This theorem shows how pseudomagic provides a complementary perspective to magic resource theory in determining the limits of magic state distillation protocols. If we have a target magic state $\ket{B}\bra{B}$ and a generic input state $\rho$, naive lower bounds from resource theory say that we require $\Omega({\mathcal{M}(\ket{B}\bra{B})}/{\mathcal{M}(\rho)})$ copies of $\rho$ to synthesize $\ket{B}\bra{B}$. This assumes that we can freely convert from nonstabilizerness in the input state to nonstabilizerness in the output state. However, once we take into account the computational efficiency of our synthesization protocol, \cref{thm:distill} intuitively means that if we do not know what the input state is, the `value' of the nonstabilizerness in the input state is reduced logarithmically. For instance, assume we have a generic resource state vector $\ket{\psi}$ with $\mathcal{M}({\psi})=O(n)$. Naive bounds tell us that we can distill at most $r= O(n)$ canonical magic states $\ket{T}$, since $\mathcal{M}(\ket{T}\bra{T}^{\otimes r}) \propto r$. However, \cref{thm:distill} imposes a \emph{much stricter} bound; if our synthesis protocol is efficient, it can synthesize at most $r = O(\log^{1+c}n)$ copies of $\ket{T}$. 

\emph{Independence of pseudomagic and pseudoentanglement: tighter distillation bounds.} The astute reader may notice that pseudomagic states are the \emph{same} states that display pseudoentanglement~\cite{aaronson_quantum_2023}. Given the long history of resource theories centered on these two quantities, it is natural to wonder if the existence of pseudomagical ensembles is in some way implied by the existence of pseudoentangled ensembles. We answer this question in the negative by showing that the entanglement and magic of a pseudorandom ensembles can be independently tuned. 
These findings suggest that neither pseudomagic nor pseudoentanglement is a generic feature of pseudorandom ensembles.%

\begin{figure}
    \centering
    \includegraphics[width=\columnwidth]{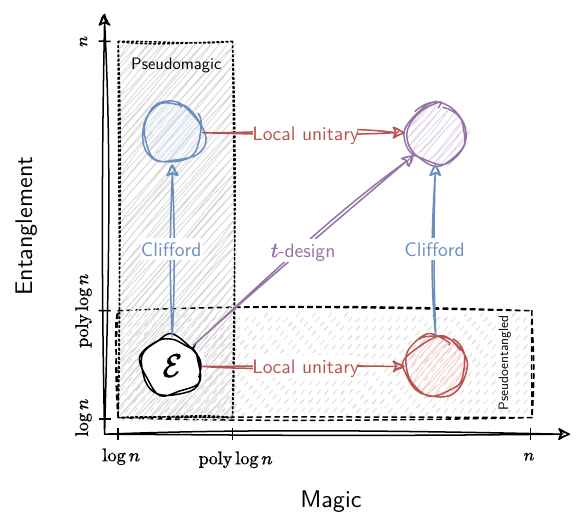}
    \caption{The pseudorandom ensemble of states $\mathcal{E}$ (subset phase states) we consider in this work exhibit both pseudomagic and pseudoentanglement. However, we show that we can \emph{independently} tune the entanglement and nonstabilizerness of these pseudorandom states via local and Clifford unitaries, respectively.}
    \label{fig:pseudotransitions}
\end{figure}

\begin{theorem}[Pseudomagic
and
pseudoentanglement are independent properties\label{thm:independence_weak}]
Given any two functions $f(n), g(n)\in[\omega(\log n),O(n)]$, there exists a pseudorandom ensemble whose states have entanglement $\Theta(f(n))$ and magic $\Theta(g(n))$ up to negligible failure probability. Moreover, there exists a pseudomagic pair with maximal gap and fixed entanglement $\Theta(f(n))$. Conversely, there exists a pseudoentangled pair with maximal gap and fixed magic $\Theta(g(n))$ (for any magic measure in Table~\ref{table:magicmonotone}).%
\end{theorem}
We study the implications of this theorem in the context of distillation protocols. \cref{thm:distill} and Proposition 3.1 of Ref.~\cite{aaronson_quantum_2023} %
are similar in spirit: they demonstrate that resource-theoretic distillability limits for EPR pairs/magic states are a gross overestimate of what is achievable by any computationally efficient algorithm \emph{that is agnostic to its input state}. However, many magic state or entanglement distillation protocols are hand-crafted to work on particular {\em classes} of input states \cite{Bravyi_2005, Bennett96}. Do our lower bounds still hold up, then, if we knew something about the state we started with? For example, Ref. \cite{Bao_2022} asks whether one can distill magic from highly entangled states. We use \cref{thm:independence_weak} to show that prior knowledge of the nonstabilizerness of input states does not lead to increased efficiency in entanglement distillation.

\begin{theorem}[Prior knowledge of magic does not help entanglement distillation]
\label{entg_distillation_main}
Consider an entanglement distillation protocol that distills EPR pairs from states drawn from an ensemble $\{\psi_k\}$. Even if we are guaranteed that the states $\psi_k$ have magic $\Theta(g(n))$ for $g(n)\in[\omega(\log n),O(n)]$ with high probability, the protocol can distill at most $O(\log^{1+c} S(\rho))$ Bell pairs with high probability, where $S(\rho)$ is the von Neumann entanglement entropy of an input state $\rho$ and $c>0$.
\end{theorem}
An identical strengthening of the black-box magic state distillation protocol in \cref{thm:distill} is also possible, using a similar proof technique (see the Supplemental Material). Our results say that for states with super-logarithmic entanglement (i.e., beyond MPS), magic distillation must be highly inefficient. Beyond strengthening distillation results, \cref{thm:independence_weak} allows us to significantly strengthen a wide array of no-go results for any entanglement (resp., nonstabilizerness) manipulation or detection task to cases where we even have a priori knowledge about states' nonstabilizerness (resp., entanglement).

{\em Conclusions and outlook.} In this work, we introduced the concept of pseudomagic states, providing an insightful expansion of the magic resource theory. We first established the theoretical foundation for pseudomagic states. The core of this framework is the stabilizer entropy, which is unique among various magic measures (see \cref{table:magicmonotone}) as it does not require an impractical minimization procedure to be evaluated either theoretically or experimentally, and lower bounds all the other magic measures. These features allowed us to demonstrate the existence of pseudomagic ensembles with a large magic gap $\log n$ vs. $n$, with respect to all the genuine magic monotones in \cref{table:magicmonotone}. We refer to the Supplemental Material for a technical discussion regarding the role of the stabilizer entropy in pseudomagic.

We investigated the implications of pseudomagic for quantum scramblers, quantum cryptography, and magic state distillation. We proved the existence of states, preparable with a non-scrambling unitary, that appear to have been generated by a scrambling one, concluding that a computationally bound observer cannot determine whether a process is scrambling solely by the generated states. More general questions concerning the existence of \textit{pseudorandom unitaries}~\cite{haug2023pseudorandom} and its consequence for quantum chaos theory will be the subject of future research. Another exciting avenue to explore is the implication that pseudomagic has for many-body physics, particularly in light of the recent studies of magic-state resource theory~\cite{oliviero_magicstate_2022,lami_quantum_2023,haug_quantifying_2023, chen_magic_2022,odavic_complexity_2022,rattacaso_stabilizer_2023,tarabunga2023manybody,chen2023magic} and quantum complexity in many-body systems \cite{Yunger_Halpern_2022}.
Finally, our exploration of Computationally Bounded Observers has revealed the subtleties of quantum phenomena under computational constraints. This perspective, guided by computational limitations, deepens our understanding of quantum systems and underscores the necessity of considering computational constraints in quantum research, opening avenues for future investigations at the intersection of quantum information and computation.

\emph{Acknowledgements.} The authors thank Anurag Anshu, J.~Pablo Bonilla, Matthias Caro, Bill Fefferman, Jonas Haferkamp, Marcel Hinsche, Marios Ioannou, Nazl\i \ U\u{g}ur K\"oyl\"uo\u{g}lu, Salvatore F.~E.~Oliviero,  Kunal Sharma, Ryan Sweke for insightful discussions, 
We also thank William Kretschmer for suggesting the proof of \cref{thm:crypto}, and Joseph Slote for first using the term ``pseudomagic" and discussing different magic measures. We thank the DFG (CRC 183, FOR 2724), the BMBF (Hybrid, RealistiQ, QSolid), the ERC (DebuQC), and the Munich Quantum Valley (K-8) for funding. SFY would like to thank the NSF for funding via the Cornell HDR institute and the CUA PFC.

\let\oldaddcontentsline\addcontentsline%
\renewcommand{\addcontentsline}[3]{}%
\medskip

\bibliographystyle{apsrev4-2}

\bibliography{references}

\let\addcontentsline\oldaddcontentsline

\onecolumngrid
\newpage

\setcounter{secnumdepth}{2}
\setcounter{equation}{0}
\setcounter{figure}{0}
\renewcommand{\thetable}{S\arabic{table}}
\renewcommand{\theequation}{S\arabic{equation}}
\renewcommand{\thefigure}{S\arabic{figure}}
\titleformat{\section}[hang]{\normalfont\bfseries}
{Supplemental Material \thesection:}{0.5em}{\centering}
\newtheorem{definitionS}{Definition S\ignorespaces}
\newtheorem{theoremS}{Theorem S\ignorespaces}
\newtheorem{corollaryS}{Corollary S\ignorespaces}

\newtheorem{lemmaS}{Lemma S\ignorespaces}
\newtheorem{claimS}{Claim S\ignorespaces}
\newtheorem{remarkS}{Remark S\ignorespaces}

\crefformat{theoremS}{Theorem #2S#1#3}
\crefformat{definitionS}{Definition #2S#1#3}
\crefformat{lemmaS}{Lemma #2S#1#3}
\crefformat{remarkS}{Remark #2S#1#3}
\crefformat{corollaryS}{Corollary #2S#1#3}

\begin{center}
\textbf{\large Supplemental Material}
\end{center}
\setcounter{equation}{0}
\setcounter{figure}{0}
\setcounter{table}{0}

\renewcommand{\theequation}{S\arabic{equation}}
\renewcommand{\thefigure}{S\arabic{figure}}
\renewcommand{\bibnumfmt}[1]{[S#1]}
\newtheorem{thmS}{Theorem S\ignorespaces}

In this supplemental material, we provide proofs, additional details supporting the claims in the main text, as well as additional applications of pseudomagic.

\section{Preliminaries}\label{sec:prelim}
There are many ways to quantify magic. In this work, we consider only \emph{monotones}; that is, those measures that satisfy the following properties:
\begin{definitionS}[Magic monotone]
    A magic monotone is a scalar function on quantum states, that has the following properties:
    \begin{itemize}
        \item $\mathcal{M}(\sigma)=0$ if and only if $\ket{\sigma}\in\Sigma$,
        \item $\mathcal{M}(\mathcal{S}(\psi))\le \mathcal{M}(\psi)$ for every stabilizer operation $\mathcal{S}$.
    \end{itemize}
\end{definitionS}
The magic measures $\mathcal{M}$ we consider in our study obey $(i)$ sub-additivity $\mathcal{M}({\psi}\otimes{\phi})\le \mathcal{M}({\psi})+\mathcal{M}({\phi})$, $(ii)$ are bounded $0\le \mathcal{M}\le n+o(1)$~\footnote{More precisely, the robustness of magic~\cite{howard_application_2017} is upper bounded by $\mathcal{R}\le n+O(2^{-n})$ as proven in Ref.~\cite{liu_manybody_2022}.}, which will be assumed throughout this work.
\subsection{Properties of stabilizer R\'{e}nyi entropies}

In this work we will focus on the measure of $\alpha$-stabilizer entropies. Let us now introduce this family of measures formally. We observe that given a $n$-qubit (pure or mixed) state $\psi$, we can define the probability distribution
\begin{equation}\label{eq:prob}
\Xi_{\psi}\coloneqq \{d^{-1}\tr^2(P\psi)\pur^{-1}(\psi) \mid P\in\mathbb{P}_n\},
\end{equation}
where $\mathbb{P}_n$ is the Pauli group on $n$ qubits and $\pur(\psi)\coloneqq \tr(\psi^2)$ its purity. The stabilizer entropy $M_{\alpha}$ is defined as (for every $0\le\alpha<\infty$)
\begin{equation}
    M_{\alpha}(\psi)\coloneqq S_{\alpha}(\Xi_{\psi})+S_{2}(\psi)-\log d ,
\end{equation}
where $S_{\alpha}(\Xi_{\psi})$ is the classical $\alpha$-R\'enyi entropy of the probability distribution in \cref{eq:prob}, while $S_{2}(\psi)\coloneqq -\log\pur(\psi)$ is the log-purity of $\psi$. To be concrete, we list three explicit expressions for $M_\alpha$ for common choices of $\alpha$ (we also assume $\psi$ is pure for simplicity),
\begin{equation}
    M_\alpha(\psi) = \begin{cases}
                         \log(\abs{\{P \in \mathbb{P}_n \mid \tr(P \psi) \neq 0 \}}) - \log d & \text{for $\alpha=0$} ,\\
                         -\sum_{P}\frac{\tr^2(P\psi)}{d}\log \tr^2(P\psi) &\text{for $\alpha=1$}, \\
                         -\log(\sum_P \frac{\tr^4(P \psi)}{d}) &\text{for $\alpha=2$.}
    \end{cases} \label{eq:malpha012}
\end{equation}
The first two cases of $\alpha=0$ and $1$ are manifestly different from the definition $M_{\alpha}({\psi})=\frac{1}{1-\alpha}\log \frac{1}{d}\sum_{P}\tr^{2\alpha}(P\psi)$ presented in Equation (1) of the main text. Despite the formal difference, the $\alpha=0$ and $1$ cases are not discontinuous exceptions to the definition in Equation (1). Rather, they follow from the definition of the classical R\'enyi entropies: $\lim_{\alpha \to 0} S_\alpha(\Xi_\psi) = \log \abs{\text{supp} \ \Xi_\psi}$, where $\text{supp} \ \Xi_\psi$ is the support of the distribution $\Xi_\psi$. Similarly, $\lim_{\alpha \to 1} S_\alpha(\Xi_\psi) = -\sum_P \Xi_\psi(P) \log(\Xi_\psi(P))$ (i.e., the Shannon entropy of $\Xi_\psi$). Therefore, despite the fact that the cases $\alpha=0$ and $\alpha=1$ may superficially seem to be special exceptions, in what follows (unless otherwise specified), we will be studying $M_\alpha$ for any $\alpha$.

Let us list important properties of the $\alpha$-stabilizer entropies.
\begin{itemize}
    \item $M_{\alpha}(\sigma)=0$ if and only if $\psi=\sum_{P\in G_{\sigma}}\phi_P P$, where $G_{\sigma}$ is a commuting subgroup of $\mathbb{P}$ and $\phi_P\in\{\pm1\}$. Notice that $\sigma$ is pure if and only if $\abs{G_{\sigma}}=2^n$. The zeros of the stabilizer entropy (i.e., the states $\psi$ for which $M_\alpha(\psi)=0$) coincides with the set of pure stabilizer state plus mixed stabilizer states defined above. However, $M_{\alpha}(\psi)>0$ when $\psi$ is a convex combination of stabilizer states: $\psi = \sum_{\sigma \in \Sigma} p_\sigma\st{\sigma}$, where $p_\sigma \geq 0$. Notice that the zeros of the stabilizer entropy do not coincide with the zeros of other magic measures for mixed states.

    \item For every Clifford unitary operator $C$, then $M_{\alpha}(C^{\dag}\psi C)=M_{\alpha}(\psi)$;
    \item $0\le M_{\alpha}(\psi)<n$, where the upper bound is a strict inequality.
    \item They feature a hierarchy: $M_{\alpha}(\psi)<M_{\alpha'}(\psi)$ for any $\psi$ if $\alpha>\alpha'$.
    \item They are additive: $M_{\alpha}(\psi\otimes \phi)=M_{\alpha}(\psi)+M_{\alpha}(\phi)$.
    \item For pure states ${\psi}$, they lower bound the stabilizer nullity: $M_{\alpha}({\psi})\le\nu({\psi})$.
    \item For pure states and $\alpha>1/2$, they lower bound the robustness of magic $M_{\alpha}({\psi})<2\mathcal{R}({\psi})$~\cite{heinrich_robustness_2019}.
    \item For pure states and $\alpha>1$, they lower bound the stabilizer fidelity: $M_{\alpha}({\psi})<\frac{2\alpha}{\alpha-1}\mathcal{F}_{\textrm{stab}}({\psi})$~\cite{haug2023stabilizer}.
    \item Defining $T$-count as smallest number $t(\psi)$ of $T$-gates used to prepare the state vector $\ket{\psi}$ from a stabilizer state, the stabilizer entropies lower bound the $T$-count: $M_{\alpha}({\psi})\le t(\psi)$.
    \item Using the swap trick, when $\alpha$ is an integer greater than 1, we can also write $M_\alpha$ as follows:
    \begin{equation}
        M_{\alpha}({\psi})=\frac{1}{1-\alpha}\log \tr(\Pi^{(2\alpha)}\psi^{\otimes 2\alpha})\label{eq:eq111},
    \end{equation}
    with $\Pi^{(2\alpha)} \coloneqq d^{-1}\sum_{P \in \mathbb{P}_n}P^{\otimes 2\alpha}$ being a Hermitian operator. Thanks to this property, stabilizer entropies can be measured on quantum computers with both single copy protocols such as randomized measurements~\cite{oliviero2022MeasuringMagicQuantum}, as well as multi-copy protocols~\cite{haug_scalable_2023,haug2023stabilizer}.
\end{itemize}

{\subsection{A remark on the role of stabilizer entropy in defining pseudomagic}}
Stabilizer entropies serve as pure-state magic monotones for every integer $\alpha\ge 2$, while they demonstrably violate the monotonicity condition for $\alpha<2$, see~\cite{leone2024stabilizermonotone}. A pure-magic monotone implies that, for deterministic stabilizer protocols—those mapping pure states to pure states—the quantity remains nonincreasing. Given that stabilizer entropies are inherently defined for pure states, this property is the natural one to consider. Furthermore, in Ref.~\cite{leone2024stabilizermonotone}, an extension of stabilizer entropy to mixed states, which adheres to monotonicity under general stabilizer protocols, is also presented.

Besides the prominent role of stabilizer entropies among magic measures (being analytically computable and experimentally measurable), the stabilizer entropy has proven to be a fundamental tool for defining and exploring pseudomagic. For instance, a key property we make frequent use of, especially when developing a general framework for pseudomagic, was that several other popular magic monotones are lower bounded by the stabilizer entropy, hence qualified as useful pseudomagic measures.

The use of stabilizer entropy for pseudomagic is of utmost importance for quantum highlights in the context of quantum scrambling theory (given its strict relationship to OTOCs) and in quantum cryptographic implications (thanks to the Fannes inequality for $M_1$). 

Moreover, stabilizer entropies play a key role in defining pseudomagic. To be more concrete, all the other monotones, except for stabilizer entropy, involve a minimization procedure that renders them computationally and experimentally intractable to measure. Consequently, they are unsuitable for proving a pseudomagic gap, which asks for measurement using multiple copies to provide the lower bound $\omega(\log n)$ for pseudomagic ensembles. Thanks to the fact that stabilizer entropy $(i)$ provides a lower bound for all the other magic monotones and $(ii)$ can be experimentally measured, we are able to demonstrate the existence of states with an extensive and \textit{tight} gap of $\log n$ versus $n$ in magic, which holds true for all the monotones discussed in this paper. This fact is reflected by \cref{cor:3-pseudomagic}.%

\subsection{Pseudorandom states}
The other key idea we use throughout this work is the notion of pseudorandom states, which are ensembles of states that are disguised as Haar random.

\begin{definitionS}[Pseudorandom states~\cite{jirandom2018}]\label{def:pseudorandom}
An ensemble of $n$-qubit quantum states $\left\{\ket{\psi_{k}} \mid k=1,\ldots,2^m\right\}$ (where $m\in\mathbb{N}$) is said to be pseudorandom, if:
\begin{itemize}
    \item Given $k$, the corresponding $\ket{\psi_{k}}$ is preparable in polynomial time.
    \item For any efficient quantum algorithm $\mathcal{A}$,
    \begin{equation}
        \abs{\Pr[\mathcal{A}\big(\ket{\psi_{k}}^{\otimes \poly(n)}\big)=1] - \Pr[\mathcal{A}\big(\ket{\phi}^{\otimes \poly(n)}\big)=1]} = \negl(n),
    \end{equation}
    with high probability for $k \sim \textnormal{Uniform}(2^m)$ and $\phi \sim \haar$. In other words, polynomially many copies of $\ket{\psi_k}$ are computationally indistinguishable from polynomially many copies of a Haar random state.
\end{itemize}
\end{definitionS}
The subset phase states %
are a class of states that
are
capable of satisfying \cref{def:pseudorandom} under certain settings. Specifically, they are ensembles of states
\begin{equation}
    \ket{\psi_{f,S}} = \frac{1}{\sqrt{\abs{S}}}\sum_{x \in S} (-1)^{f(x)} \ket{x}.
\end{equation}
indexed by a function $f: \{0,1\}^n \to \{0,1\}$ and subsets $S \subseteq \{0,1\}^n$.
If $\abs{S} = 2^{\omega(\log n)}$ and $f$ is a truly random function, then the subset phase states are statistically indistinguishable from the Haar random states (i.e., polynomially many copies of $\ket{\psi_{f,S}}$ are close in trace distance to polynomially many copies of a Haar random state)~\cite{aaronson_quantum_2023}. However, if $f$ is randomly chosen, the corresponding subset phase states are not efficiently preparable. If, on the other hand, $f$ and $S$ are chosen pseudorandomly (using quantum-secure pseudorandom functions and subsets), $\ket{\psi_{f,S}}$ \emph{are} efficiently preparable. Furthermore, although the resulting ensemble of states are statistically distinguishable from Haar random states, they are \emph{computationally indistinguishable}. That is, when $f$ and $S$ are chosen pseudorandomly and $\abs{S} = 2^{\omega(\log n)}$, $\ket{\psi_{f,S}}$ qualify as pseudorandom states under \cref{def:pseudorandom}~\cite{aaronson_quantum_2023}.

\subsubsection*{Experimental feasibility} 
We propose a potential method that uses Rydberg arrays to realize the subset phase states $\ket{\psi_{f,S}}$. This illustrates that pseudomagic is not merely an artificial theoretical construction: it can be implemented in the laboratory. The first step of our construction provides a concrete implementation for the black-box oracles $f$ described in Ref.~\cite{aaronson_quantum_2023}. Assuming $\abs{S} = 2^k$ ($k \leq n$), for any function $f: \{0,1\}^k \to \{0,1\}$, one can associate a hypergraph over $k$ sites with hyperedges $E$ such that $f(x) = \sum_{\{i_1,\ldots,i_m\} \in E} \prod_{\alpha=1}^{m}x_{i_{\alpha}}$. One can attach the appropriate phases to the bit string states by simply applying multi-qubit CZ gates on sites are connected by hyperedges \cite{liu2022many} as
\begin{equation}
\prod_{\{i_1,\ldots,i_m\} \in E} C^{m-1} Z_{i_1 \ldots i_m} \ket{+}^{\otimes k} = \sum_{x \in \{0,1\}^k} (-1)^{f(x)} \ket{x}.
\end{equation}
Finally, by concatenating $\ket{0}^{\otimes (n-k)}$ and applying a pseudorandom permutation of computational basis states~\cite{aaronson_quantum_2023}, we arrive at the desired $\ket{\psi_{f,S}}$. This protocol has quasipolynomial time complexity, as the hypergraph representation of $f$ can have a quasipolynomial number of edges, so we anticipate that for near-term qubit counts $n$, the construction of these states will be feasible. In particular, neutral-atom platforms such as Rydberg arrays, are very well-suited for implementing this construction for two key reasons. First, these platforms offer the advantage of being able to quickly move qubits around, allowing for the realization of the required hypergraph connectivity. Secondly, these platforms natively support the implementation of multi-qubit CZ gates \cite{levine2019parallel}. This streamlines the creation of subset phase states $\ket{\psi_{f,S}}$ using a native set of gates, making the implementation straightforward and noise-resilient. 

\subsection{A remark about the need for computational assumptions}
Having defined the subset phase states that have been introduced in Ref.~\cite{aaronson_quantum_2023}, we are in position to highlight the computational assumptions made in this paper and their implications in the various theorems presented herein. While Ref.~\cite{aaronson_quantum_2023} showed that subset phase states can be prepared efficiently assuming the existence of quantum-secure pseudorandom functions, we observe most of our applications (except for cryptography) do not require efficient preparability of pseudomagic states. For all such applications, we can instantiate them with pseudomagic states that exist {\em unconditionally} (i.e. without need for computational complexity assumptions): simply pick $f$ and $S$ to be a truly random function and subset, respectively, in which case the subset phase states are statistically indistinguishable from Haar random states (Theorem 2.1 of Ref.~\cite{aaronson_quantum_2023}).

\section{Bounds to pseudomagic}\label{sec:pseudomagicproofs}

\subsection{Proof of Lemma 1}\label{subsec:stabentropyfact1}
First, we will state a result that lower bounds the magic of any state ensemble that is computationally indistinguishable from Haar random states. This can be of independent interest to the reader. This implies a lower bound on the magic of any state ensemble that is computationally indistinguishable from a high magic ensemble, as we remark after the proof. This is because the only property we need of Haar random states, for the distinguisher to work, is that they have high magic.

\begin{lemmaS}[Bounds to R\'enyi entropies (Lemma 1, repeated)]%
    \label{lem:stab-bound-rep}
    Let $\mathcal{E}$ be an ensemble of pseudorandom states (i.e., computationally indistinguishable from Haar random states, which is an ensemble with high magic). For $\ket{\psi} \sim \mathcal{E}$,  one has that with high probability
    \begin{equation}
        M_{\alpha}({\psi})=\begin{cases}
                               \omega\left(\frac{\log n}{\alpha-1}\right) \ &\text{for odd $\alpha$},\\
                               \omega\left(\frac{\log n}{\alpha}\right) \ &\text{for even $\alpha$},
        \end{cases}
    \end{equation}
    for any $\alpha=O(\poly n)$.
\end{lemmaS}

\begin{proof}
    It follows from the definition of R\'enyi entropies that for any $\ket{\psi}$, the exponential of the R\'enyi entropy is the expectation value 
    \begin{equation} \label{eq:renU}
    \tr(\Pi^{(2\alpha)}\psi^{\otimes 2\alpha}) = 2^{(1-\alpha) M_{\alpha}({\psi})}
    \end{equation}
    of the Hermitian unitary operator $\Pi^{(2\alpha)}=d^{-1}\sum_{P}P^{\otimes 2\alpha}$.
    From \cref{lem:levy}, we see on the one hand that
    \begin{equation}\label{eq:Haar}
    \Pr_{\ket{\psi}\sim \haar}[M_{\alpha}({\psi}) \geq n\beta/\alpha]  \approx \Pr_{\ket{\psi}\sim \haar}[\tr(\Pi^{(2\alpha)}\psi^{\otimes 2\alpha}) \leq 2^{-n\beta}] \geq 1-2^{-C_1/\alpha^2 2^{(1-2\beta)n}}.
    \end{equation}
    On the other hand, for $\ket{\psi} \sim \mathcal{E}$, assume towards contradiction that with high probability
    \begin{equation}\label{eq:PRS}
    \tr(\Pi^{(2\alpha)}\psi^{\otimes 2\alpha}) = \Omega\left(\frac{1}{\poly(n)}\right).
    \end{equation}
    The difference in the values of $\tr(\Pi^{(2\alpha)}\psi^{\otimes 2\alpha}) $ depending on whether $\psi$ is Haar-random or from $\mathcal{E}$ provides a route to proving our lower bound on $M_{\alpha}$ for pseudorandom states.

    Let us set $\alpha$ to be odd first. Then we will show there exists an efficient distinguisher $\mathcal{A}$ that, given $O (\alpha \poly(n))$ copies, distinguishes whether an unknown $\ket{\psi}$ is drawn Haar-randomly or from $\mathcal{E}$. This distinguisher builds on an observation of Ref.~\cite{haug_efficient2023}, where the authors show that to measure $\tr(\Pi^{(2\alpha)}\psi^{\otimes 2\alpha})$ (with odd $\alpha$), one simply needs to run a Hadamard test with the unitary $\Pi^{(2\alpha)}$ and estimate the acceptance probability. However, we don't even need to estimate $\tr(\Pi^{(2\alpha)}\psi^{\otimes 2\alpha})$; our efficient distinguisher $\mathcal{A}$ simply runs the above-mentioned Hadamard test with $k=2\alpha$ copies of $\psi$ and outputs 1 if the test outputs 1, which happens with probability
    \begin{equation}
\frac{1+\tr(\Pi^{(2\alpha)}\psi^{\otimes 2\alpha})}{2}.
    \end{equation}
    Then from \cref{eq:Haar,eq:PRS},
    \begin{equation}
        \begin{gathered}
            \abs{\Pr_{\ket{\psi}\sim \mathcal{E}}[\mathcal{A}(\ket{\psi}^{\otimes k}) = 1]  - \Pr_{\ket{\psi}\sim \haar}[\mathcal{A}(\ket{\psi}^{\otimes k}) = 1]}         \\
            > \frac{1+\Omega\left(\frac{1}{\poly(n)}\right)}{2}-\left(1-2^{-C_1/\alpha^2 2^{(1-2\beta)n}}\right)\left(\frac{1+2^{-n\beta}}{2}\right) - O\left(2^{-2^n}\right) \\
            = \Omega\left(\frac{1}{\poly(n)}\right)    ,
        \end{gathered}
    \end{equation}
    thus contradicting the definition of pseudorandom quantum states. Therefore, we conclude that for odd $\alpha$,
    \begin{equation}
        \tr(\Pi^{(2\alpha)}\psi^{\otimes 2\alpha})=o\left(\frac{1}{\poly(n)}\right)
    \end{equation}
    and from \cref{eq:renU} and \cref{eq:PRS} one has that for a state drawn from $\mathcal{E}$ and odd $\alpha$,
    \begin{equation}
        M_{\alpha}({\psi})=\omega((\alpha-1)^{-1}\log n)
    \end{equation}
    For the case of even $\alpha$, it is sufficient to employ the hierarchy of stabilizer entropies $M_{\alpha}\geq M_{\alpha+1}$ to conclude that
    \begin{equation}
        M_{\alpha}({\psi_{f,S}})=\omega(\alpha^{-1}\log n).
    \end{equation}

\end{proof}

\begin{remark}
    The only property we needed of Haar random states in the proof of \cref{lem:stab-bound} was that they have stabilizer R\'enyi entropy $\Omega(n)$ (see \cref{eq:Haar}). Hence, the same distinguisher, as sketched in \cref{lem:stab-bound}, works to distinguish \emph{any} states with stabilizer R\'enyi entropy $\Omega(n)$ from states with stabilizer R\'enyi entropy $O(\log n)$. Therefore, \cref{lem:stab-bound}, in full generality, follows the same proof as \cref{lem:stab-bound}.
\end{remark}

\subsection{Stabilizer R\'enyi entropy as a pseudomagic measure}\label{subsec:lem1proof}

A key result of this work is that the stabilizer entropy is a good candidate for defining the notion of pseudomagic.
\begin{remark}[Bounds to non-integer stabilizer entropies]
    All our conclusions and theorems presented below apply only to integer values of $\alpha$ where $\alpha \geq 2$. The reason behind this limitation is that for non-integer $\alpha$ or $\alpha=0,1$, we are unable to utilize the swap trick for computing Haar averages. Consequently, the proofs provided below do not hold in such cases. However, using the hierarchy of stabilizer entropies
    \begin{equation}
        M_{\alpha}({\psi}) \geq M_{\alpha'}({\psi})\quad \forall \alpha'>\alpha,
    \end{equation}
    we can bound every non-integer $\alpha$ from both above and below if we wish.
\end{remark}

The following lemmas characterize the stabilizer entropies attained by Haar random states and pseudorandom states.

\begin{lemmaS}[Stabilizer entropies for Haar random states]\label{lem:Haar}
For every integer $\alpha\ge 2$ and for $d = 2^n$, the Haar average $\alpha$-stabilizer entropy is lower bounded as
\begin{equation}
    \E_{\ket{\psi} \sim \haar}[M_\alpha(\psi)]\ge \begin{cases} n-2+O(d^{-1}) &\text{for $\alpha=2$},
    \\\frac{n}{\alpha-1}+O(d^{-1}) &\text{for $\alpha\geq 3$}.
    \end{cases}
\end{equation}
\begin{proof}
    The proof can be found in Ref.~\cite{leone_clifford_}.
\end{proof}
\end{lemmaS}

Notice that since we are considering integer $\alpha\ge2$, $\frac{1}{\alpha-1}>0$, so one has due to Jensen inequality
\begin{equation}
    \E_{\ket{\psi} \sim \haar}[M_\alpha(\psi)]\ge -\frac{1}{\alpha-1}\log \E_{\ket{\psi} \sim \haar}\left[\tr(\Pi^{(2\alpha)} \psi^{\otimes 2\alpha})\right].
\end{equation}
Levy's lemma then allows us to convert this statement about expectated stabilizer entropies to a statement in probability.

\begin{lemmaS}[Typicality]\label{lem:levy}
Let $\alpha$ be an integer such that $\alpha\ge2$, and let $\beta<1/2$. If $\ket{\psi}$ is a random Haar state, then
\begin{equation}
    M_{\alpha}({\psi})\ge\begin{cases}\frac{n\beta}{\alpha-1} & \text{for odd $\alpha$}\\
    \frac{n\beta}{\alpha} & \text{for even $\alpha$}
    \end{cases}
\end{equation}
with probability $\ge 1-2^{-\frac{Cd^{1-2\beta}}{4(\alpha+1)^2}}$, where $C=O(1)$.
\begin{proof}
    Consider the $\alpha$-stabilizer entropy for odd integer $\alpha$. Let us compute an upper bound for the Lipschitz constant of the function $\tr(\Pi^{(2\alpha)}\psi^{\otimes 2\alpha})$ as
    \begin{equation}
        \abs{\tr(\Pi^{(2\alpha)}\psi^{\otimes 2\alpha})-\tr(\Pi^{(2\alpha)}\phi^{\otimes 2\alpha})}\le \norm{\Pi^{(2\alpha)}}_{\infty}\norm{\psi^{\otimes 2\alpha}-\phi^{\otimes2\alpha}}_{1}\le 2\alpha\norm{\ket{\psi}-\ket{\phi}}
    \end{equation}
    where we have used the fact that $\Pi^{(2\alpha)}$ is unitary and thus $\norm{\Pi^{(2\alpha)}}_{\infty}=1$.
    We have also made use of the
    following identity multiple times
    \begin{equation}
        \begin{aligned}
            \norm{\psi^{\otimes 2\alpha}-\phi^{\otimes2\alpha}}_{1}&=\norm{\psi^{\otimes 2\alpha}-\psi\otimes\phi^{\otimes 2\alpha-1}+\psi\otimes\phi^{\otimes 2\alpha-1}-\phi^{\otimes2\alpha}}_{1}\\
            &\le \norm{\psi}_{1}\norm{\psi^{\otimes 2\alpha-1}-\phi^{\otimes 2\alpha-1}}+\norm{\ket{\psi}-\ket{\phi}}\norm{\phi^{\otimes 2\alpha-1}}_1\\
            &=
            \norm{\psi^{\otimes 2\alpha-1}-\phi^{\otimes 2\alpha-1}}+\norm{\ket{\psi}-\ket{\phi}}. \label{eqn:repeat-distance}
        \end{aligned}
    \end{equation}
    Note that
    \begin{equation}
        \norm{\psi-\phi}_{1}^2=\norm{\ket{\psi}-\ket{\phi}}^2={1-\abs{\braket{\psi}{\phi}}^2}.
    \end{equation}
    Therefore, using
    Levy's lemma,
    we can write
    \begin{equation}
        \Pr_{\ket{\psi} \sim \haar}\left[\abs{\tr(\Pi^{(2\alpha)}\psi^{\otimes 2\alpha})-d^{-1}+O(d^{-2})}\ge \epsilon\right]\le 2^{-\frac{Cd\epsilon^2}{4\alpha^2}}.
    \end{equation}
    Choosing $\epsilon\coloneqq d^{-\beta}$ with $\beta<1/2$, we can write
    \begin{equation}
        \Pr_{\ket{\psi} \sim \haar}\left[\tr(\Pi^{(2\alpha)}\psi^{\otimes 2\alpha})\le d^{-\beta}+O(d^{-1})\right]\ge 1-2^{-\frac{Cd^{1-2\beta}}{4\alpha^2}}.
    \end{equation}
    Therefore, using \cref{eq:eq111} and %
    exploiting the fact that $d
    \coloneqq 2^n$, we arrive
    at
    \begin{equation}
        \Pr_{\ket{\psi} \sim \haar}\left[M_{\alpha}({\psi})\ge \frac{\beta n}{\alpha-1}+O(d^{-1})\right]\ge 1-2^{-\frac{Cd^{1-2\beta}}{4\alpha^2}},
    \end{equation}
    which proves the statement for odd integer $\alpha\ge 3$. For even $\alpha$, is sufficient to note that the stabilizer entropies follow
    the hierarchy
    \begin{equation}
        M_{\alpha}({\psi})\ge M_{\alpha+1}({\psi}).
    \end{equation}
    Therefore, with probability greater than $1-2^{-\frac{Cd^{1-2\beta}}{4(\alpha+1)^2}}$, one has
    \begin{equation}
        M_{\alpha}({\psi})\ge M_{\alpha+1}({\psi})\ge\frac{\beta n}{(\alpha+1)-1}=\frac{\beta n}{\alpha}.
    \end{equation}
\end{proof}
\end{lemmaS}

We now upper bound the stabilizer entropy of subset phase states in particular, by upper-bounding $M_0$, as shown in the following lemma.
\begin{lemmaS}[Upper bound to $M_{0}$ for subset phase states]\label{lem:upperboundonm0}
Any subset phase state vector $\ket{\psi_{f,S}}$ satisfies
\begin{equation}
    M_{0}({\psi_{f,S}})\le 2\log \abs{S}.
\end{equation}
\begin{proof}
    Recall the definition of $M_0$ as
    \begin{equation}
        M_0(\psi) = \log(\abs{\{P \in \mathbb{P}_n \mid \tr(P \psi) \neq 0 \}}) - \log d.
    \end{equation}
    In order to derive the bound, we need to upper bound the number of Pauli operators that can have component on $\ket{\psi_{f,S}}$, i.e., the cardinality of the set $\text{supp} \ \Xi_\psi \coloneqq  \{P \in \mathbb{P}_n \mid \tr(P \psi) \neq 0\}$. First of all, note that every Pauli operator $P$ can be written as a product $P=\phi P_X P_Z$, where $\phi \in\{\pm 1,\pm i\}$, $P_X\in\mathbb{X}_n$, $P_{Z}\in\mathbb{Z}_n$, and $\mathbb{X}_n,\mathbb{Z}_n$ are the commuting groups generated by the single qubit $X$ and $Z$s, respectively.
    Consider the expectation value 
    \begin{equation}
        \mel{\psi_{f,S}}{P_X}{\psi_{f,S}}=\frac{1}{\abs{S}}\sum_{x,x'\in S}(-1)^{f(x)+f(x')}\mel{x}{P_X}{x'}
    \end{equation}
    of $P_X \in \mathbb{X}_n$ over $\ket{\psi_{f,S}}$.
    Now, let us consider fixing the sum over $x'$. Since $P_X\ket{x'}$ corresponds to a computational basis state, it follows that at most $\abs{S}$ of these states can belong to the set $S$, allowing them to potentially have a non-zero overlap $\mel{x}{P_X}{x'}$. Given that there are $\abs{S}$ possible $\ket{x'}$ states on which $P_X$ can act, a (loose) upper bound on the number $N_X$ of Pauli operators belonging to the subgroup $\mathbb{X}_n$ with a non-zero expectation value is $N_X \leq \abs{S}^2$. Now, consider a product some $P_X$ with a non-zero expectation value and any element $P_Z \in \mathbb{Z}_n$. Since $\abs{\mel{x}{P_X P_{Z}}{x'}} = \abs{\mel{x}{P_X}{x'}}$, it is possible that every $P_Z \in \mathbb{Z}_n$ could contribute to a non-zero expectation value. Therefore, the cardinality of the set $\text{supp} \ \Xi_\psi$ is upper bounded by
    \begin{equation}
        \abs{\text{supp} \ \Xi_\psi}\le N_{X} \abs{\mathbb{Z}_n} \le \abs{S}^{2}d.
    \end{equation}
    Therefore, we conclude that $M_{0}({\psi_{f,S}})\le 2\log \abs{S}$.
\end{proof}
\end{lemmaS}

\begin{theoremS}[Stabilizer entropies of pseudorandom subset phase states]\label{thm:malpha-prs}
The $\alpha$-stabilizer entropy of pseudorandom subset phase states with size-$\abs{S}$ subsets satisfies
\begin{equation}
    \omega(\alpha^{-1} \log n) < M_{\alpha}(\psi_{f,S}) \leq O(\log \abs{S}).
\end{equation}
\begin{proof}
    The lower bound follows from \cref{lem:stab-bound}. The upper bound follows from the bound $ M_{0}({\psi_{f,S}})\le 2\log \abs{S}$ proven in \cref{lem:upperboundonm0}, and from the fact that $M_{\alpha}\le M_0$ for every $\alpha>0$. Specifically, for pseudorandom subset phase states with $\abs{S} = \Theta(\exp \poly\log n)$, $M_\alpha(\ket{\psi_{f,S}}) = O(\poly \log n)$.
\end{proof}
\end{theoremS}

\begin{corollaryS}[Stabilizer entropies as a pseudomagic measure]\label{cor:stabilizerentropy}
As a corollary, we have that every stabilizer entropy $M_{\alpha}({\psi})$ with $\alpha=O(\poly\log n)$ is a suitable magic measure to provide a notion of pseudomagic.

\begin{proof}
    Subset phase states are pseudorandom for $\abs{S}=2^{\omega(\log n)}$ and $\abs{S}=2^{O(\poly\log n)}$, therefore, with non-negligible probability one has $M_{\alpha}({\psi_{f,S}})=O(\poly\log n)$, while from \cref{lem:Haar}, Haar random states have $M_{\alpha}({\psi_{\haar}})= O(n/\alpha)$ with high probability.
\end{proof}
\end{corollaryS}

\subsection{Tight magic bounds for subset phase states}\label{subsec:tight-bound}
Below, we prove a useful bound on $M_2({\psi_{f}})$ for random phase states. A phase state vector, where the phases are given by the function $f$, is the state vector
\begin{equation}
    \ket{\psi_f} = \frac{1}{\sqrt{2^n}} \sum_{x \in \{0, 1\}^n} (-1)^{f(x)} \ket{x}.
\end{equation}
This is equivalent to setting $S = \{0,1\}^n$ for the subset phase states (and so from now on we drop the subscript $S$ in the notation). We anticipate its usefulness for proving various concentration inequalities relating to $M_2$. The techniques used in the following proof will also be used for later results (i.e., \cref{thm:lowerboundm2} and \cref{thm:tune-magic}) that improve upon \cref{thm:malpha-prs} by finding \emph{tight} magic bounds for subset phase states.

\begin{lemmaS}[Average on R\'enyi-2 stabilizer entropies]\label{lem:expm2}
If $f$ is sampled uniformly at random from all possible binary functions $f: \{0,1\}^n \to \{0,1\}$, then
\begin{equation}
    \E_{f}\left[2^{-M_{2}({\psi_{f}})}\right]\le \frac{20}{2^n}.
\end{equation}

\begin{proof}
    By definition of $M_2$, we have
    \begin{equation}
        \begin{aligned}
            \E_{f}\left[2^{-M_2({\psi_{f}})}\right] & = \frac{1}{2^n} \sum_i \tr\left(P_i^{\otimes 4} \E_f\left[\ketbra{\psi_{f}}{\psi_{f}}^{\otimes 4}\right]\right)                                                         \\
            & = \frac{1}{32^n}\sum_{i,\{x_j\},\{y_j\}} \E_f\left[(-1)^{\sum_j f(x_j)+f(y_j)}\right] \tr(P_i^{\otimes 4} \ket{x_1,x_2,x_3,x_4} \bra{y_1,y_2,y_3,y_4}), \label{eqn:trpi4}
        \end{aligned}
    \end{equation}
    where $\{x_j\}$ denotes the fact that we are summing over all possible sets of four length-$n$ bitstrings $x_1,x_2,x_3,x_4$ (similarly for $\{y_j\}$). Now, observe that there are two constraints that must be simultaneously satisfied in order for the summand to be nonzero.
    \crefname{enumi}{constraint}{constraints}
    \Crefname{enumi}{Constraint}{Constraints}
    \crefname{fact}{fact}{facts}
    \Crefname{fact}{Fact}{Facts}

    \crefname{case}{case}{cases}
    \Crefname{case}{Case}{Cases}

    \begin{enumerate}[label=(\alph*)]
        \item $\E_f\left[(-1)^{\sum_j f(x_j)+f(y_j)}\right]=1$ only when all the $\{x_j\}$ and $\{y_j\}$ can be `paired up'. That is, any given bitstring can only appear an even number of times in $\{x_j\} \cup \{y_j\}$. \label[enumi]{constr:even}
        \item For every $P_i$, we must have $P_i \ket{x_j} \propto \ket{y_j}$ for $j=1,\ldots,4$. \label[enumi]{constr:pi}
    \end{enumerate}
    To upper bound the sum \cref{eqn:trpi4}, we consider six cases for the summands. We simply need to count the number of nonzero summands, since the magnitude of each of the summands is at most 1.
    \begin{enumerate}[label=\Roman*.]
        \item\label[case]{itm:sum1} $P_i \in \mathbb{Z}_n$, where $\mathbb{Z}_n$ is the group generated by the Pauli $\sigma^z$s. Then \Cref{constr:pi} amounts to $x_j = y_j$, and then \Cref{constr:even} is automatically satisfied. The sum then reads
        \begin{equation}
            \sum_{P \in \mathbb{Z}_n} \sum_{\{x_j\}} \tr(P^{\otimes 4} \ketbra{x_1,x_2,x_3,x_4}{x_1,x_2,x_3,x_4})=\sum_{P \in \mathbb{Z}_n} \tr(P^{\otimes 4})=1
        \end{equation}
        (the only nonzero term is from the identity). For the remaining cases, we can then assume $P_i \not\in \mathbb{Z}_n$, which amounts to assuming $x_j \neq y_j$.
        \item\label[case]{itm:sum2} The $x_j$ are all the same ($x_1=x_2=x_3=x_4$). There are at most $2^n$ choices for $\{x_j\}$ and $4^n$ choices for $P_i$, so $\leq 8^n$ summands that fall under this case.
        \item\label[case]{itm:sum3} Three of the $x_j$ are the same, and one is different. There are $\binom{4}{3}=4$ ways to choose this triplet, and $4^n$ choices for $\{x_j\}$ (with $4^n$ choices for $P_i$), amounting to at most $\leq 4 \cdot 16^n$ nonzero summands.
        \item\label[case]{itm:sum4} There are two distinct pairs amongst $x_j$ (for instance, $(x_1=x_3)\neq(x_2=x_4)$). There are at most $\binom{4}{2}=6$ ways to choose these pairings, and $4^n$ choices for $\{x_j\}$ after these pairings are chosen (and at most $4^n$ choices for $P_i$). This case represents $\leq 6 \cdot 16^n$ nonzero summands.
        \item\label[case]{itm:sum5} Two of the $x_j$ are the same, and two are different. For concreteness, say $x_1=x_2\neq x_3\neq x_4$. To satisfy \Cref{constr:even}, we must have $(x_3,x_4)=(y_3,y_4)$ or $(x_3,x_4)=(y_4,y_3)$ (note we automatically have $y_1=y_2$ by \Cref{constr:pi}). However, recalling that $P_i \not\in \mathbb{Z}_n$ (hence $\ket{x_j} \perp \ket{y_j}$), we cannot have $(x_3,x_4)=(y_3,y_4)$, and must, therefore, have $(x_3,x_4)=(y_4,y_3)$. Combining this with \Cref{constr:pi}, we get $P_i \ket{x_3} \propto \ket{y_3}=\ket{x_4}$. Therefore, we only have freedom in choosing one of $x_1$ or $x_2$, and in choosing $x_3$. Generalizing this, we find a total of $\binom{4}{2}$ choices of pairings in the $\{x_j\}$, $4^n$ choices for $\{x_j\}$, and $4^n$ choices for $P_i$. This represents at most $6 \cdot 16^n$ summands.
        \item\label[case]{itm:sum6} All of the $x_j$ are unique. In this case, to satisfy \Cref{constr:even}, the sum must look like
        \begin{equation}
            \sum_{i, \{x_j\},\sigma \in \mathbb{S}_4} \tr(P_i^{\otimes 4} \ket{x_1,x_2,x_3,x_4}\bra{x_{\sigma(1)},x_{\sigma(2)},x_{\sigma(3)},x_{\sigma(4)}}),
        \end{equation}
        where $\mathbb{S}_4$ is the group of permutations over $1,\ldots,4$. \Cref{constr:pi} then reads $P_i \ket{x_j} = \ket{x_{\sigma(j)}}$. Applying this twice, we have $\ket{x_j} = P_i^2 \ket{x_j} = \ket{x_{\sigma(\sigma(j))}}$, and since the $\{x_j\}$ are unique, this implies $\sigma(\sigma(j))=j$ (i.e., $\sigma$ is an involution). Furthermore, since $\ket{x_j} \perp \ket{x_{\sigma(j)}}$ (by assumption that $P_i \not\in \mathbb{Z}_n$), $\sigma$ must act nontrivially on each of $1,\ldots,4$. These two facts imply that $\sigma$ must be a product of two disjoint two-cycles (e.g., it swaps $1 \leftrightarrow 3$ and $2 \leftrightarrow 4$). Therefore, we only have freedom in choosing two of the $\{x_j\}$ (one for each of the two-cycles), and then the other two are uniquely determined. To count the number of nonzero terms here, there are only $3$ permutations $\sigma$ that are a product of disjoint two-cycles, at most $4^n$ choices for $\{x_j\}$, and at most $4^n$ choices for $P_i$. This amounts to $\leq 3 \cdot 16^n$ nonzero summands.
    \end{enumerate}
    In summary, we have at most $1+8^n+(3+6+4+6) \cdot 16^n \leq 20 \cdot 16^n$ summands. Therefore,
    \begin{equation}
        \E_{f}\left[2^{-M_2({\psi_f})}\right] \leq \frac{20}{2^n}.
    \end{equation}
\end{proof}
\end{lemmaS}

Note that in the proof of \cref{lem:expm2}, the only property of $f$ that we have used has been in \Cref{constr:even}. This constraint holds so long as $f$ is 8-wise independent, and does not require $f$ to be completely random. We can construct pseudorandom 8-wise independent functions as follows. We randomly choose $h$ from the 8-wise independent family
\begin{equation}
    H = \{h \mid \{0,1\}^n \to \{0,1\}\},
\end{equation}
and $q$ from a family of quantum-secure permutations
\begin{equation}
    Q = \{q \mid \{0,1\}^n \to \{0,1\}^n\}.
\end{equation}
Then, defining $f \coloneqq  h \circ q$, $f$ is both pseudorandom and 8-wise independent. This has been proven in full detail in Ref.~\cite{aaronson_quantum_2023}. This gives us the following theorem.

\begin{theoremS}[Tight bounds on the stabilizer entropies of subset phase states]\label{thm:lowerboundm2}
Let $\omega(\log n) \leq k \leq n$ and $\abs{S} = 2^k$. If $f$ is sampled from the ensemble of 8-wise independent pseudorandom functions, then the associated subset phase states satisfy
\begin{equation}
    M_\alpha({\psi_{f,S}}) = \Theta(k)
\end{equation}
with high probability for any $\alpha \in [0,2]$.
\end{theoremS}
\begin{proof}
    The upper bound follows from \cref{thm:malpha-prs}. For the lower bound, we focus on the case $\alpha=2$, and define
    \begin{equation}
        Z_f \coloneqq  -\log(M_2(\ket{\psi_{f,S}})) = \frac{1}{2^n} \sum_{P \in \mathbb{P}_n} \tr(P^{\otimes 4} \ketbra{\psi_{f,S}}{\psi_{f,S}}^{\otimes 4}),
    \end{equation}
    for an arbitrary fixed subset $S$ (with $\abs{S}=2^k$). Our goal will be to show that
    \begin{equation}
        \mathbb{E}_f[Z_f] = O(\abs{S}^{-1}),
    \end{equation}
    regardless of the choice of subset $S$. In general, do this by mirroring the techniques of \cref{lem:expm2}. However, handling \Cref{itm:sum1} requires some special care. We will show that
    \begin{equation}
        \frac{1}{\abs{S}^4 2^n} \sum_{P_Z \in \mathbb{Z}_n} \left(\sum_{x \in S} \tr(P \ketbra{x}{x})\right)^4 = O(\abs{S}^{-1}) \label{eq:z-case}
    \end{equation}
    We represent $P_Z \in \mathbb{Z}_n$ as a bitstring $z$, where $z_i=1$ if $P$ has a $Z$ on the $i$th position, so the sum over $P_Z \in \mathbb{Z}_n$ then turns into a sum over all $z \in \{0,1\}^n$ We also then have $g_z(x) \coloneqq  \tr(P \ketbra{x}{x}) = (-1)^{x \cdot z}$, where $x \cdot z$ is shorthand for $\sum x_i z_i$. By expanding out the fourth power of the sum over $x \in S$, we will get a collection of terms $g_z(x_i)g_z(x_j)g_z(x_k) g_z(x_l)$, which can be simplified to simply $g_z(x_i+x_j+x_k+x_l)$, where the addition is mod 2. Then, we note that when we take the sum over all $z \in \{0,1\}^n$, this term averages to zero, unless $x_i+x_j+x_k+x_l=0^n$ (i.e., a string of all zeros). In order for this requirement to hold, we see that we have at most $O(\abs{S}^3)$ choices for $x_i,x_j,x_k$, after which $x_l$ is uniquely determined. In summary, there are at most $O(\abs{S}^3)$ terms that do not sum to zero, and so we have shown \cref{eq:z-case}.

    The remaining cases are simple. Since \Cref{constr:even,constr:pi} still hold, we can run through \Crefrange{itm:sum2}{itm:sum6} of \cref{lem:expm2} making adjustments as needed. The main adjustment is that instead of $4^n$ choices for $P_i$, there are $2^n \abs{S}$ choices. The reasoning for this follows that of \cref{lem:upperboundonm0}: picking a $P_i$, is equivalent to picking a $P_X \in \mathbb{X}_n$ and a $P_Z \in \mathbb{Z}_n$. There are at most $2^n$ choices for $P_Z$, and only $\abs{S}$ choices for $P_X$, since there are at most $\abs{S}$ different bitstrings to which we can flip to when we are working with subset phase states. For \Cref{itm:sum2}, there are $\abs{S}$ choices for $\{x_j\}$, and for each of these choices, there are at most $\abs{S} 2^n$ choices of $P_i$ so that all of the corresponding $\{y_i\}$ belongs to $S$. For \Crefrange{itm:sum3}{itm:sum6}, there are $\abs{S}^2$ choices for $\{x_j\}$ and $\abs{S} 2^n$ choices for $P_i$.
    Therefore, we have that
    \begin{equation}
        Z =\frac{1}{2^n} \sum_i \left(P_i^{\otimes 4} \E\left[\st{\psi_{f,S}}^{\otimes 4}\right]\right) \leq \underbrace{O\left(\frac{1}{\abs{S}}\right)}_{\textrm{\Cref{itm:sum1}}} + \underbrace{\frac{\abs{S}^2 2^n + (3+6+4+6) \cdot \abs{S}^3 2^n}{2^n \abs{S}^4}}_{\textrm{\Crefrange{itm:sum2}{itm:sum6}}} \leq O\left(\frac{1}{\abs{S}}\right).
    \end{equation}
    For this reason, there must be some $N_0$ for which $N \geq N_0 \implies Z \leq \frac{1}{\sqrt{\abs{S}}} = 2^{-k/2}$. By virtue of Markov's inequality, we find
    \begin{equation}
        \Pr_{f}\left[Z_f \geq 2^{-k/4}\right] \leq 2^{-k/4},
    \end{equation}
    and it follows that
    \begin{equation}    \Pr_{\ket{\psi}_{f,S}}\left[M_2({\psi_{f,S}}) \leq \frac{k}{4} \right] \leq 2^{-k/4}.
    \end{equation}
    As an immediate corollary, when $f$ is sampled from the ensemble of $8$-wise independent pseudorandom functions, $M_\alpha({\psi_{f,S}}) = \Theta(k)$ with high probability for any $\alpha \in [0,2]$.
\end{proof}

\subsection{Robustness of magic as a pseudomagic measure}\label{subsec:lemmarobustness}

Recall the definition of the log-robustness of magic $\mathcal{R}({\rho})$ as
\begin{equation}
    \mathcal{R}({\rho})\coloneqq \log\min\left\{\norm{c}_1 \big\vert \ \rho=\sum_{\ket{\sigma} \in \Sigma}c_\sigma\st{\sigma}\right\}.
\end{equation}
Note that this definition of $\mathcal{R}({\rho})$ is the logarithm of the conventional definition given in
Ref.~\cite{heinrich_robustness_2019}.

\begin{theoremS}[Robustness of subset phase states]\label{thm:robustness}
The robustness of subset phase states satisfies, with high probability over the ensemble
\begin{equation}
    \mathcal{R}({\psi_{f,S}}) = \Theta(\log \abs{S}).
\end{equation}
Moreover, we can state a more precise upper bound $\mathcal{R}(\psi_{f,S}) \leq \log \abs{S}$.
\begin{proof}
    The lower bound follows from combining \cref{thm:lowerboundm2} with the fact that the robustness of magic is lower bounded by $M_2$ to get $\mathcal{R}({\psi_k}) \geq M_2(\psi_k) = \Theta(\log \abs{S})$. We will now prove the upper bound by bounding the robustness of $\st{\psi_{f,S}}$. First, note that
    \begin{equation}
        \st{\psi_{f,S}}=\frac{1}{\abs{S}}\sum_{i,j}(-1)^{f(x_i)+f(x_j)}\ketbra{x_i}{x_j}=\frac{1}{\abs{S}}\sum_{i}\st{x_i}+\frac{1}{\abs{S}}\sum_{i=1}^{\abs{S}}\sum_{j=i+1}^{\abs{S}}(-1)^{f(x_i)+f(x_j)}(\ketbra{x_i}{x_j}+c.c.).
    \end{equation}
    Now, we have that $\st{x_i}$ are stabilizer states for every $x_i\in S$ and that
    \begin{equation}
        \ketbra{x_i}{x_j}+c.c=\left(\frac{\ket{x_i}+\ket{x_j}}{\sqrt{2}}\right)\left(\frac{\bra{x_i}+\bra{x_j}}{\sqrt{2}}\right)-\left(\frac{\ket{x_i}-\ket{x_j}}{\sqrt{2}}\right)\left(\frac{\bra{x_i}-\bra{x_j}}{\sqrt{2}}\right) \eqqcolon \sigma_{i,j}^{+}-\sigma_{i,j}^{-},
    \end{equation}
    where $\sigma_{i,j}^{\pm}$ are stabilizer states. Therefore, we can express the density matrix for phase states
    \begin{equation}
        \st{\psi_{f,S}}=\frac{1}{\abs{S}}\sum_{i}\st{x_i}+\frac{1}{\abs{S}}\sum_{i=1}^{\abs{S}}\sum_{j=i+1}^{\abs{S}}(-1)^{f(x_i)+f(x_j)}(\sigma_{i,j}^{+}-\sigma_{i,j}^{-}).
    \end{equation}
    From the above expression we can easily upper bound the robustness of magic as
    \begin{equation}
        \mathcal{R}(\psi_{f,S})\le \log\left(\frac{\abs{S}}{\abs{S}}+\frac{2}{\abs{S}}\sum_{i=1}^{\abs{S}}\sum_{j=i+1}^{\abs{S}} 1\right) = \log\left(1+ \frac{2}{\abs{S}} \frac{\abs{S}(\abs{S}-1)}{2}\right)= \log \abs{S}.
    \end{equation}
\end{proof}
\end{theoremS}
\begin{corollaryS}[Robustness of magic as a pseudomagic measure]\label{cor:robustness}
The robustness of magic is a good candidate for defining the notion of pseudomagic.
\end{corollaryS}

\subsection{Sufficient conditions for pseudomagic measures}\label{subsec:sufficient}
We first show that a sufficient criterion for a magic measure $\mathcal{M}$ to be a good pseudomagic measure is that $\mathcal{M}$ is lower bounded by a stabilizer entropy and upper bounded by the robustness of magic.

\begin{theoremS}[Sufficient conditions for bounded
good pseudomagic measures]\label{thm:sufficient}
Let $\mathcal{M}({\psi})$ be a magic monotone for pure states. Then, if  $\mathcal{M}({\psi})$ is bounded as $\Omega(M_\alpha({\psi}))\leq \mathcal{M}({\psi}) \leq O(\mathcal{R}({\psi}))$ for some constant $\alpha$, then $(i)$ $\mathcal{M}=\omega(\log n)$ for any ensemble of pseudorandom states; $(ii)$ $\mathcal{M}$ is a measure of pseudomagic with maximum gap $\Omega(n)$ vs.~$O(\poly\log n)$.
\begin{proof}
    From \cref{lem:stab-bound}, we have that for every ensemble of pseudorandom states $\mathcal{E}$, $\mathcal{M}({\psi_k})=\omega(\log n)$ with high probability over the choice of $\ket{\psi_k}\in\mathcal{E}$. From \cref{thm:robustness}, we have that for pseudorandom subset phase states $\mathcal{M}({\psi_k})=O(\log \abs{S})$ and thus for $\abs{S}=2^{O(\poly\log n)}$, $\mathcal{M}({\psi_k})=O(\poly\log n)$, satisfying the definition of pseudomagic.
\end{proof}
\end{theoremS}
Then, we have the following two corollaries that readily follow from \cref{thm:sufficient}.

\begin{corollaryS}[Stabilizer fidelity, stabilizer extent, and max-relative entropy]\label{cor:3-pseudomagic}
Each of the following magic measures fulfill the conditions in the definition of pseudomagic, hence they can be considered good pseudomagic measures with gap $\Theta(n)$ vs. $\Theta(\poly \log n)$.
\begin{enumerate}[label=(\roman*)]
    \item\label{itm:item1} the stabilizer fidelity $\mathcal{F}_{\mathrm{stab}}({\psi})\coloneqq -\log\max_{\ket{\sigma} \in \Sigma}\abs{\braket{\sigma}{\psi}}^2$,
    \item\label{itm:item2} stabilizer extent $\xi({\psi})\coloneqq \log\min\left\{\sum_\phi \abs{c_{\phi}}^2 \mid \ket{\psi} = \sum_{\ket{\phi} \in \Sigma} c_\phi \ket{\phi} \right\}$, and
    \item\label{itm:item3} the max relative entropy of magic $D_{\max}(\ket{\psi})=\log\min\{\lambda \mid \lambda\sigma-\psi\geq 0\}$, where $A \geq 0$ means $A$ is positive semidefinite (when $A$ is a matrix).
\end{enumerate}
\begin{proof}
    These magic measures obey bounds with respect to the stabilizer R\'enyi entropy $M_2$ and the robustness of magic $\mathcal{R}$ that have already been established.
    \begin{equation}
        \frac{1}{4} M_2(\psi) < \mathcal{F}_{\textrm{stab}}(\psi) \leq \xi(\psi), D_{\textrm{max}}(\psi) \leq \mathcal{R}(\psi),
    \end{equation}
    where the first inequality was shown in Ref. \cite{haug2023stabilizer} and the last two were shown in Ref. \cite{liu2022many}.
    The tightness of the gap $\Theta(n)$ vs. $\Theta(\poly \log n)$ (as opposed to $\Omega(n)$ vs. $O(\poly \log n)$) is due to the fact that when $\psi$ is a pseudorandom subset phase state, we have already shown that $M_2(\psi) = \Theta(\log \abs{S})$ in \cref{thm:lowerboundm2} and $\mathcal{R}(\psi) = \Theta(\log \abs{S})$ in \cref{thm:robustness}. Furthermore, we know $M_2(\psi) = \Omega(n)$ and $\mathcal{R}(\psi) = \Omega(n)$ for Haar random states, but since both of these measures are upper bounded by $n+o(1)$ anyways, they can trivially be rewritten $M_2(\psi) = \Theta(n)$ and $\mathcal{R}(\psi)=\Theta(n)$.
\end{proof}
\end{corollaryS}

\subsection{Stabilizer nullity is a poor candidate to define pseudomagic}\label{subsec:stabnull}
In this section, we show that the stabilizer nullity is a poor candidate to define pseudomagic. More precisely, we show that pseudorandom states must necessarily have maximal stabilizer nullity, mirroring the behavior observed in Haar-random states. To start, we formally define stabilizer nullity, first introduced in Ref.~\cite{beverland_lower_2020} (see also Ref.~\cite{jiang_lower_2021}). Following the notation of Ref.~\cite{leone2023learningdopedstates}, for any $\ket{\psi}$, one can associate a subset $G_{\psi}$ of the Pauli group $\mathbb{P}_n$ comprised of Pauli operators that have unit expectation value on $\psi$:
\begin{equation}
    G_{\psi}\coloneqq \{P\in\mathbb{P}_n \mid \tr(P\psi)=\pm1\}.
\end{equation}
It is easy to show that $G_{\psi}$ is a group~\cite{jiang_lower_2021}, so we say $G_{\psi}$ is the \emph{stabilizer group} associated with the state $\psi$. For stabilizer states $\sigma$, the cardinality of $G_{\sigma}$ is maximal: $|G_{\sigma}|=d$. Conversely, for non stabilizer states $\psi$ one has $|G_{\psi}|<d$. Therefore, one can define a magic measure called the \emph{stabilizer nullity}
\begin{equation}
    \nu(\psi)\coloneqq n-\log\abs{G_{\psi}}.
\end{equation}
We are now ready to present the outcome of this section, contained in the following lemma.
\begin{lemmaS}\label{lem:stabnullity}
Any ensemble $\mathcal{E}=\{\ket{\psi_k}\}$ of computationally pseudorandom quantum states must obey $\nu(\psi_k)=n$ with high probability over $k$.
\end{lemmaS}
\begin{proof}
    The proof of this lemma follows easily from the fundamental results of Ref.~\cite{grewal2023improved}: it shows that an ensemble of pseudorandom quantum states $\mathcal{E}$ must feature $|G_{\psi_k}|=1$ with high probability over the choice of $k$. The proof is based on two key facts: $(i)$ Haar random states obey $\log\abs{G_{\psi}}=0$ with overwhelming probability; $(ii)$ the construction of an efficient algorithm that is able to discriminate whether $\log\abs{G_{\psi_k}}>1$ or $\log\abs{G_{\psi_k}}=1$. A simple corollary is that the stabilizer nullity for the ensemble $\mathcal{E}$ of pseudorandom states is $\nu(\psi_k)=n$ with high probability over the choice of $k$.
\end{proof}
This lemma tells us that the stabilizer nullity is unlikely to be a suitable magic measure for defining pseudomagic, as no ensemble of pseudorandom quantum states can exhibit a nonzero gap in stabilizer nullity compared to Haar-random states. Nevertheless, due to the generality of the definition for psedomagic compared to pseudorandom states, the possibility remains open that there are pseudomagic states, which do not meet the criteria of computational pseudorandomness, yet exhibit a substantial pseudomagic gap. This remains an open question for future investigations.

Beyond this, \cref{lem:stabnullity} has broader implications for the relationship between stabilizer nullity and other magic measures (particularly, the robustness of magic). \cref{cor:3-pseudomagic} establishes a sufficient condition for magic monotones $\mathcal{M}$ to exhibit a substantial pseudomagic gap, assuming that $\mathcal{M}$ is bounded by a certain stabilizer entropy $M_{\alpha}$ and the robustness $\mathcal{R}$. As mentioned in \smref{sec:prelim}, the stabilizer nullity is lower bounded by $M_{\alpha}(\psi)\le \nu(\psi)$ for any $\alpha$. On the other hand, there exists no known corresponding upper bound involving the robustness of magic. Notably, simply by combining \cref{cor:3-pseudomagic,lem:stabnullity}, we arrive at the following corollary, which rules out the possibility of having any upper bound on the stabilizer nullity in terms of the robustness of magic.
\begin{corollaryS}
    Let $\nu(\psi)$ be the stabilizer nullity and $\mathcal{R}(\psi)$ be the robustness of magic defined in Table 1 of the main text. Then, there is no constant $a$ such that $\nu(\psi)\le a\mathcal{R}(\psi)$ for every state $\psi$.
\end{corollaryS}

\section{The continuity of pseudorandomness}\label{sec:prs-cont}

\begin{lemmaS}[Indistinguishability of ensembles]\label{lem:error-ensemble}
For any two ensembles of states $\mathcal{E} \coloneqq  \{\psi_k\}$ and $\mathcal{E}' \coloneqq  \{\phi_k\}$, if
\begin{equation}
    \TV(\psi_k,\phi_k) < o\left(\frac{1}{\poly n}\right)
\end{equation}
for all $k$, then $\mathcal{E}$ and $\mathcal{E}'$ are statistically (hence computationally) indistinguishable.
\end{lemmaS}
\begin{proof}
    This follows from the fact that for any $t = O(\poly n)$, $\TV(\psi_k^{\otimes t}, \phi_k^{\otimes t}) \leq t \cdot \TV(\psi_k, \phi_k) < o\left(\frac{1}{\poly n}\right)$. Then by a simple triangle inequality,
    \begin{equation}
        \TV\left(\E_{\psi_k \sim \mathcal{E}}[\ketbra{\psi_k}{\psi_k}^{\otimes t}], \E_{\phi_k \sim \mathcal{E}}[\ketbra{\phi_k}{\phi_k}^{\otimes t}]\right) \leq \E\left[\TV(\ketbra{\psi_k}{\psi_k}^{\otimes t}, \ketbra{\phi_k}{\phi_k}^{\otimes t})\right] < o\left(\frac{1}{\poly n}\right).
    \end{equation}
\end{proof}
\begin{lemmaS}[Transitivity of computational indistinguishability]\label{lem:transitivty}
For any ensembles $\mathcal{E}_1, \mathcal{E}_2, \mathcal{E}_3$, if $\mathcal{E}_1 \cong \mathcal{E}_2$ and $\mathcal{E}_2 \cong \mathcal{E}_3$ (where $\cong$ denotes computational indistiguishability), then $\mathcal{E}_1 \cong \mathcal{E}_3$.
\end{lemmaS}
\begin{proof}
    This follows from a simple triangle inequality. For any efficient algorithm $\mathcal{A}$ and any $t = O(\poly n)$,
    \begin{equation}
        \begin{gathered}
            \abs{\Pr_{\psi_k \sim \mathcal{E}_1}[\mathcal{A}\big(\ket{\psi_{k}}^{\otimes t}\big)=1] - \Pr_{\psi_k \sim \mathcal{E}_3}[\mathcal{A}\big(\ket{\psi_k}^{\otimes t}\big)=1]} \\
            \leq \abs{\Pr_{\psi_k \sim \mathcal{E}_1}[\mathcal{A}\big(\ket{\psi_{k}}^{\otimes t}\big)=1] - \Pr_{\psi_k \sim \mathcal{E}_2}[\mathcal{A}\big(\ket{\psi_k}^{\otimes t}\big)=1]} + \abs{\Pr_{\psi_k \sim \mathcal{E}_2}[\mathcal{A}\big(\ket{\psi_{k}}^{\otimes t}\big)=1] - \Pr_{\psi_k \sim \mathcal{E}_3}[\mathcal{A}\big(\ket{\psi_k}^{\otimes t}\big)=1]} \\
            \leq \negl(n).
        \end{gathered}
    \end{equation}
\end{proof}
Combining \cref{lem:error-ensemble} and \cref{lem:transitivty} allows us to conclude that the pseudorandomness of an ensemble is robust to to small state preparation errors. More precisely, if we have a target ensemble of pseudorandom states $\mathcal{E} = \{\psi_k\}$, but over the course of preparing these states, we make small errors such that the resulting ensemble $\mathcal{E}'=\{\psi_k'\}$ satisfies
\begin{equation}
    \TV(\psi_k, \psi_k') < o\left(\frac{1}{\poly n}\right),
\end{equation}
for all $k$, then $\mathcal{E}'$ is still pseudorandom. This is because \cref{lem:error-ensemble} guarantees $\mathcal{E} \cong \mathcal{E}'$, which, combined with \cref{lem:transitivty}, tells us $\mathcal{E} \cong \haar \implies \mathcal{E}' \cong \haar$.

\section{Quantum chaos: Proof of Theorem 2}

\label{sec:chaosproof}
\begin{definitionS}[Quantum chaos]\label{defchaos}
Consider a unitary $U$. Its $2k$-point out-of-time-order correlators (OTOCs)~\cite{roberts_chaos_2017,hosur_chaos_2016} are
\begin{equation}
    C_{2k}(U)\coloneqq \frac{1}{d}\tr(\tilde{P}_1 Q_1\tilde{P}_2 Q_2\cdots\tilde{P}_k Q_k),
\end{equation}
where $P_i$ and $Q_i$ are non-identity Pauli operators for $i=1,\ldots, k$ and $\tilde{P}_i\coloneqq  U^{\dag}P_i U$. We say $U$ is chaotic if
\begin{equation}
    C_{2k}(U)=\tilde{O}(C_{2k}(U_{\haar})),
\end{equation}
where $\tilde{O}$ stands for an irrelevant \emph{polynomial overhead}, and $C_{2k}(U_{\haar})$ is the average of the $2k$-point OTOCs over Haar random unitaries.
\end{definitionS}

\begin{theoremS}[Exponential separation from Haar values]
    Let $\mathcal{E}$ be an ensemble of pseudomagic states. Let $\ket{\psi}\in\mathcal{E}$ and let $U$ such that $\ket{\psi}=U\ket{0}^{\otimes n}$. Then there exists $2k$-point OTOCs that are exponentially separated from the Haar value as
    \begin{equation}
        C_{2k}(U)=\Omega(\exp(n)C_{2k}(U_{\haar}))
    \end{equation}
    and thus $U$ cannot be considered chaotic by \cref{defchaos}.
    \begin{proof}
        The proof will be constructive. We simply construct a $2k$-point OTOC that exhibits exponential separation from the Haar average value.
        Let $U$ and $\ket{\psi_U}$ be unitary and state vector such that $\ket{\psi_U}=U\ket{0}^{\otimes n}$. Consider the $\alpha$-stabilizer entropy $M_{\alpha}(
        \psi_U
        )$ with integer $\alpha>1$
        \begin{equation}
            M_{\alpha}({\psi_U})=\frac{1}{1-\alpha}\log \frac{1}{d}\sum_{P \in \mathbb{P}_n}\tr^{2\alpha}(\tilde{P}\st{0}^{\otimes n})
        \end{equation}
        where $\tilde{P}\coloneqq U^{\dag}PU$. Let us exploit the fact that $\st{0}^{\otimes n}$ is pure to write $\tr^2(\tilde{P}\st{0}^{\otimes n})=\tr(\tilde{P}\st{0}^{\otimes n}\tilde{P}\st{0}^{\otimes n})$ multiple times. Expanding $\st{0}^{\otimes n}$ in terms of Pauli operators as $\st{0}^{\otimes n}=d^{-1}\sum_{Z\in\mathbb{Z}_n}Z$ with $\mathbb{Z}_n$ the commuting subgroup of the Pauli group generated by single-qubit Pauli-$Z$s, we write
        \begin{equation}
            M_{\alpha}({\psi_U})=\frac{1}{1-\alpha}\log\frac{1}{d^{2\alpha+1}}\sum_{P,Z_1,\ldots, Z_{2\alpha}}\tr(\tilde{P}Z_1\tilde{P}Z_{2}\cdots\tilde{P}Z_{2\alpha}).
        \end{equation}
        Now, let us denote by brackets $\langle\cdot\rangle_{\mathbb{P}_n,\mathbb{Z}_n}$ the group average with respect to the whole Pauli group $\mathbb{P}_n$ and the subgroup $\mathbb{Z}_n$. Noting that $C_{4\alpha}(U,P,Z_i)\coloneqq  d^{-1}\tr(\tilde{P}Z_1\tilde{P}Z_{2}\cdots\tilde{P}Z_{2\alpha})$ is a $4\alpha$-point OTOC, let us define $C_{4\alpha}(U)\coloneqq  \langle C_{4\alpha}(U,P,Z_i)\rangle_{\mathbb{P}_n,\mathbb{Z}_n}$ the averaged $4\alpha$ point OTOC. We can thus express the stabilizer entropy in terms of a single out of time ordered correlation function
        \begin{equation}
            M_{\alpha}({\psi_U})=\frac{1}{1-\alpha}\log d^2 C_{4\alpha}(U).
        \end{equation}
        Now we are finally ready to show the main result. Let $\ket{\psi_{PM}}$ be a pseudomagic state satisfying $M_{\alpha}({\psi_{PM}})=o(n)$ (the subset phase states with $\abs{S}=O(\exp \poly\log n)$ easily accomplish this). Let $\ket{\psi_{\haar}}$ be a Haar random state. With overwhelming probability, one has $M_{\alpha}({\psi_{\haar}})=\Omega(n/\alpha)$. Let us consider the difference
        \begin{equation}
            M_{\alpha}({\psi_{\haar}})-M_{\alpha}({\psi_{PM}})=\Omega(n/\alpha).
        \end{equation}
        Denoting $U_{PM}$ the unitary that constructs $\ket{\psi_{PM}}=U_{PM}\ket{0}^{\otimes n}$ and $U_{\haar}$ the equivalent for $\ket{\psi_{\haar}}$, we readily get
        \begin{equation}
            C_{4\alpha}(U_{PM})=\Omega(\exp(n)C_{4\alpha}(U_{\haar})).
        \end{equation}
        Defining $2k=4\alpha$ proves the statement.
    \end{proof}
\end{theoremS}

\section{Fannes inequality for \texorpdfstring{$M_1$}{M1} and cryptography from pseudomagic: Proof of Theorem 3}\label{sec:crypto}

Informally speaking, an EFI pair is a pair of efficient quantum algorithms whose output states are statistically far but computationally indistinguishable.

\begin{definitionS}[EFI pairs \cite{brakerski2022computational}]
    We call $\mu\coloneqq \left(\mu_{b, \lambda}\right)$ a pair of EFI states if it satisfies the following criteria:
    \begin{itemize}
        \item \emph{Efficient generation}: There exists a uniform QPT quantum algorithm $A$ that on input $\left(1^\lambda, b\right)$ for some integer $\lambda$ and $b \in\{0,1\}$, outputs the mixed state $\mu_{b, \lambda}$.
        \item \emph{Statistically distinguishable}: $\TV\left(\mu_{0, \lambda}, \mu_{1, \lambda}\right) \geq \Omega(\frac{1}{\poly \lambda})$.
        \item \emph{Computational indistinguishability}: $\left(\mu_{0, \lambda}\right)_\lambda$ is computationally indistinguishable to $\left(\mu_{1, \lambda}\right)_\lambda$.
    \end{itemize}
\end{definitionS}

We will now show that pseudomagic ensembles imply EFI pairs. A critical ingredient in this proof will be a Fannes-type inequality which establishes for stabilizer entropy what the Fannes-Audenaert inequality establishes for entropy -- namely its continuity in the space of states over the metric of trace distance. We now prove such an inequality for $M_1$. Since this proof is of independent interest, we devote a separate section to it before we prove the implication about EFI pairs.

\subsection{A Fannes Inequality for the Stabilizer Shannon Entropy}
Recall from \cref{eq:malpha012} that $M_1$ can be written
\begin{equation}
    M_{1}(\psi)=-\sum_{P \in \mathbb{P}_n}\frac{\tr^2(P\psi)}{d}\log \tr^2(P\psi).
\end{equation}
\begin{theoremS}[Fannes-like inequality for $M_1$ (also to appear in Ref. \cite{leone_clifford_})]\label{thm:Fannes}
For any two states $\ket{\psi}$ and $\ket{\phi}$,
it holds that
\begin{equation}\label{eq:Fannes}
\abs{M_{1}(\psi)-M_{1}(\phi)}\le\begin{cases} \norm{\psi-\phi}_1\log(d^2-1)+H_{\textrm{bin}}[\norm{\psi-\phi}_1] &\text{for $\norm{\psi-\phi}_1\le 1/2$},\\
\norm{\psi-\phi}_1\log(d^2-1)+1 &\text{for $\norm{\psi-\phi}_1> 1/2$.}\\
\end{cases}
\end{equation}
\end{theoremS}
\begin{proof}
    We make use of two main ingredients
    \begin{itemize}
        \item
        Consider a state vector $\ket{\psi}$ and its conjugate (in the computational basis) denoted as $\ket{\psi^{*}}$ and defined as $\ket{\psi^{*}}=\sum_{i}\braket{i}{\psi}^{*}\ket{i}$. Define $\ket{P}:= I\otimes P\ket{I}$ where $\ket{I}=\frac{1}{\sqrt{d}}\sum_{i}\ket{i}\otimes \ket{i}$ the Bell state. Note that $\{\ket{P}\}$ is a basis for the Hilbert space $\mathcal{H}^{\otimes 2}$. The dephased state vector $\ket{\psi}\otimes \ket{\psi^*}$ (which we denote $\ket{\psi \otimes \psi^*}$ for short) in the basis $\ket{P}$ equals
        \begin{equation}
            \mathcal{D}_{P}(\psi\otimes\psi^*):= \sum_{P}\st{P}\psi\otimes\psi^*\st{P}=\sum_{P}\frac{\tr^2(P\psi)}{d}\st{P}\label{eq:dephased}
        \end{equation}
        indeed $\mel{P}{\psi\otimes\psi^*}{P}=\abs{\braket{\psi\otimes\psi^*}{P}}^2$. And
        \begin{equation}
            \braket{\psi\otimes\psi^*}{P}=\frac{1}{\sqrt{d}}\sum_{i}\braket{\psi^*}{i}\mel{\psi}{P}{i}=\frac{1}{\sqrt{d}}\sum_{i}\mel{\psi}{P}{i}\braket{i}{\psi}=\frac{\mel{\psi}{P}{\psi}}{\sqrt{d}}.
        \end{equation}
        \item Moreover, let us show that $\norm{\psi\otimes\psi^*-\phi\otimes\phi^*}_1=\norm{\psi^{\otimes 2}-\phi^{\otimes 2}}_1$. To see this, note that
        \begin{equation}
            \norm{\psi\otimes\psi^*-\phi\otimes\phi^*}_1=2\sqrt{1-\tr(\psi\otimes \psi^{*}\phi\otimes \phi^*)}=2\sqrt{1-\tr(\psi\phi)\tr (\psi^{*}\phi^*)}
        \end{equation}
        then consider $\tr(\psi^*\phi^*)=\abs{\braket{\psi^*}{\phi^*}}^2$ and write it in the computational basis
        \begin{equation}
            \braket{\psi^*}{\phi^*}=\sum_{i} \braket{\psi^*}{i} \braket{i}{\phi^*}=\sum_{i} \braket{\phi}{i} \braket{i}{\psi} = \braket{\phi}{\psi}
        \end{equation}
        which implies $\tr(\psi^*\phi^*)=\tr(\psi\phi)$. We used the definition of conjugate state, i.e. $\ket{\psi^*}=\sum_{i}(\braket{i}{\psi})^*\ket{i}\equiv\sum_{i}\braket{\psi}{i}\ket{i}$. Returning to the $1$-norm distance, we have
        \begin{equation}
            \norm{\psi\otimes\psi^*-\phi\otimes\phi^*}_1=2\sqrt{1-\tr(\psi\phi)\tr (\psi^{*}\phi^*)}=2\sqrt{1-\tr^2(\psi\phi)}=\norm{\psi^{\otimes 2}-\phi^{\otimes 2}}_1.\label{eq:equalityfiga}
        \end{equation}
    \end{itemize}
    Consider now two pure states $\psi,\phi$ and the Shannon entropies $S(\Xi_{\psi}),S(\Xi_{\phi})$ of the (purely classical) probability distributions $\Xi_{\psi}(P),\Xi_{\phi}(P)$. The classical 
    \emph{Fannes-Audenaert inequality} for the Shannon entropy tells us
    \begin{equation}
        \abs{S(\Xi_{\psi})-S(\Xi_{\phi})}\le \TV(\Xi_{\psi},\Xi_{\phi})\log(d^2-1)+H_{\rm bin}[\TV(\Xi_{\psi},\Xi_{\phi})],
        \label{eq:fannesclassical}
    \end{equation}
    where $\TV(\mathbf{p},\mathbf{q})\coloneqq \frac{1}{2}\sum_{i}\abs{p_i-q_i}$ is the total variation distance between two probability distributions, and $H_{\rm bin}[x]$ is the binary entropy function $H_{\rm bin}[\epsilon]=-\epsilon \log \epsilon - (1-\epsilon) \log(1-\epsilon)$. Consider the state in \cref{eq:dephased} and notice that
    \begin{equation}
        \TV(\Xi_{\psi},\Xi_{\phi})=\frac{1}{2}\sum_{P}\abs{\Xi_{\psi}(P)-\Xi_{\phi}(P)}=\frac{1}{2}\norm{\mathcal{D}_{P}(\psi\otimes\psi^*)-\mathcal{D}_{P}(\phi\otimes\phi^*)}_1
    \end{equation}
    simply represents the distance between the dephased states in the Bell basis. We can now combine \cref{eq:equalityfiga} with the fact that the $1$-norm is non-increasing under any quantum channel (in this case, $\mathcal{D}_P$) to bound
    \begin{equation}
        \TV(\Xi_{\psi},\Xi_{\phi})\le\frac{1}{2}\norm{\psi^{\otimes 2}-\phi^{\otimes 2}}_1\le  \norm{\psi-\phi}_1.\label{eq:tv-td}
    \end{equation}
    Assuming that $\norm{\psi-\phi}_1\le 1/2$ then $\TV(\Xi_{\psi},\Xi_{\phi})\le 1/2$ and we have that $H_{\rm bin}$ is monotonic and we can write $H_{\rm bin}(\TV(\Xi_{\psi},\Xi_{\phi}))\le H_{\rm bin}(\norm{\psi-\phi}_1)$ and thus we can write
    \begin{equation}
        \abs{M_{1}(\psi)-M_{1}(\phi)}\le \norm{\psi-\phi}_1\log(d^2-1)+H_{\rm bin}[\norm{\psi-\phi}_1]
        .
    \end{equation}
    Otherwise, if $\norm{\psi-\phi}_1>1/2$, we can bound $H_{\rm bin}(\TV(\Xi_{\psi},\Xi_{\phi}))\le 1$
    \begin{equation}
        \abs{M_{1}(\ket{\psi})-M_{1}(\ket{\phi})}\le \norm{\psi-\phi}_1\log(d^2-1)+1.
    \end{equation}
\end{proof}
Following an identical technique to the above theorem, we can prove a much looser bound for higher $\alpha$.
\begin{corollaryS}[Fannes-like inequality for $M_{\alpha > 1}$]\label{cor:fannes2}
For any $\alpha > 1$ and any two state vectors $\ket{\phi}$ and $\ket{\psi}$,
\begin{equation}
    \abs{M_\alpha(\psi) - M_\alpha(\phi)} \leq \frac{d^2 \alpha}{\alpha-1} \norm{\psi-\phi}_1.
\end{equation}
\end{corollaryS}
\begin{proof}
    There are analogies to the Fannes-Audenaert inequality in \cref{eq:fannesclassical} for classical R\'enyi entropies. In Eq. (20) of Ref. \cite{hanson2022}, it is shown that
    \begin{equation}
        \abs{S_\alpha(\mathbf{p}) - S_\alpha(\mathbf{q})} \leq \dim(\mathbf{p}) \cdot \frac{\alpha}{\alpha-1}  \TV(\mathbf{p},\mathbf{q}),
    \end{equation}
    for any $\alpha>1$ and any classical distributions $\mathbf{p},\mathbf{q}$ defined over $\dim(\mathbf{p})$ elements. In our case, this simply means
    \begin{equation}
        \abs{S_\alpha(\Xi_\psi) - S_\alpha(\Xi_\phi)} \leq \frac{d^2 \alpha}{\alpha-1}\TV(\Xi_\psi, \Xi_\phi) \leq \frac{d^2 \alpha}{\alpha-1} \norm{\psi-\phi}_1,
    \end{equation}
    where the factor $d^2$ is because the classical distributions $\Xi_\psi, \Xi_\phi$ are defined over all $d^2$ Pauli matrices. The second inequality follows from \cref{eq:tv-td}.
\end{proof}

\subsection{Proof of Theorem 3}
\begin{theoremS}[EFI pairs from pseudo-magic ensembles (Theorem 3, repeated)]
    Efficiently-generatable pseudo-magic ensembles with stabilizer 1-entropy that can be tuned between $\omega(\log n)$ and $n$ imply EFI pairs. Namely, we consider as the pair an algorithm that generates a random pseudo-magic state with magic of $\poly \log n$, and one that generates a random pseudo-magic state with magic of $n$.
\end{theoremS}

\begin{proof}
    By a triangle inequality (since both ensembles are computationally indistinguishable from a Haar-random state), the outputs of these two algorithms are computationally indistinguishable from each other. It remains to show they are statistically distinguishable.

    Let $\ket{\psi}$ be sampled from the first ensemble, which has magic $\Theta(\poly\log n)$ and $\ket{\phi}$ be sampled from the second ensemble, which has magic $\Theta(n)$. To show that the two ensembles are statistically distinguishable, it suffices to show that the two are far in trace distance, that is,
    \begin{equation}
        \norm{\psi-\phi}_1 \geq \Omega \left(\frac{1}{\poly n} \right).
    \end{equation}
    The continuity of the stabilizer Shannon entropy tells us that states which have very different magic must be far in trace distance. More specifically, \cref{thm:Fannes} gives us
    \begin{equation}
        \norm{\psi - \phi}_1 \log(2^{2n}-1) \geq \Omega(n - \poly\log n),
    \end{equation}
    which is to say that
    \begin{equation}
        \norm{\psi - \phi}_1 \geq \Omega\left(1 - \frac{\poly\log n}{n}\right) \geq \Omega\left(\frac{1}{\poly n}\right).
    \end{equation}
\end{proof}

\section{Implications to magic-state distillation: Proof of Theorem 4}\label{sec:distillproof}

An immediate application of pseudorandom state ensembles with tunable stabilizer R\'{e}nyi entropy is to derive new upper bounds on the problem of magic state distillation via stabilizer operations. 

\begin{theoremS}[Bounds to magic state distillation]\label{MSD}
Let $\mathcal{M}$ be a (sub-additive) magic monotone such that $\Omega(M_{\alpha})\le \mathcal{M}\le O(\mathcal{R})$ for some $\alpha=\Theta(1)$. Then any algorithm that uses stabilizer operations to synthesize a target quantum state $\ket{B}$ from copies of an arbitrary input state $\rho$ requires
\begin{equation}
    \Omega\left(\frac{\mathcal{M}(\ket{B}\bra{B})}{\log^{1+c}(\mathcal{M}(\rho))}\right)\label{eq:mag-distill}
\end{equation}
copies of the input to succeed with non-negligible probability, where $c>0$.
\end{theoremS}

\begin{proof}
    Assume towards contradiction that there is an algorithm that uses only a number
    \begin{equation}
        o\left(\frac{\mathcal{M}(\ket{B})}{\log^{(1+c)}(\mathcal{M}(\rho))}\right)
    \end{equation}
    of copies of the (here possibly mixed) input $\rho$ to synthesize a single copy of the state $\ket{B}$, or equivalently, per copy of $\rho$, at least
    \begin{equation}
        \Omega\left(\frac{\log^{(1+c)}(\mathcal{M}(\rho))}{\mathcal{M}(\ket{B}\bra{B})}\right)\label{eqn:synth-assumption}
    \end{equation}
    copies of $\ket{B}$ may be synthesized.

    As detailed in
    Ref.~\cite{leone_stabilizer_2022}, a necessary condition to be able to synthesize a target quantum state $\ket{B}$ from an initial quantum state $\rho$ is that ${\cal M}(\rho) > {\cal M}(\ket{B}\bra{B})$ for any measure ${\cal M}$ of nonstabilizerness. This condition imposes that, per copy of $\rho$, at most
    \begin{equation}\label{eq:UB}
    O\left(\frac{\mathcal{M}(\rho)}{\mathcal{M}(\ket{B}\bra{B})}\right)
    \end{equation}
    copies of $\ket{B}$ can be synthesized, where we have used the assumed sub-additivity of $\mathcal{M}$, i.e., $\mathcal{M}(\rho^{\otimes r})\le r \mathcal{M}(\rho)$. 

    Recall our assumption that $\Omega(M_\alpha) \leq \mathcal{M} \leq O(\mathcal{R})$, which implies that there exists positive constants $n_0,s,s'$ such that
    \begin{equation}
        s M_{\alpha}(\psi) \leq \mathcal{M}(\psi) \leq s' \mathcal{R}(\psi) \label{eq:s-mbound}
    \end{equation}
    for all $n$-qubit states $\psi$ when $n>n_0$. For the rest of this proof, we will  assume $n>n_0$. To prove our statement on the limitations of a distillation protocol that works on arbitrary input states $\rho$ (either pure or mixed), it is sufficient to focus on two special classes of input states, namely Haar random states and subset phase states. By combining \cref{lem:levy,eq:s-mbound}, we find that
    \begin{equation}
        \Pr_{\ket{\psi} \sim \haar}\left[\mathcal{M}(\psi)\leq \frac{s \beta n}{\alpha}\right] \leq 2^{-\frac{Cd^{1-2\beta}}{4(\alpha+1)^2}},
    \end{equation}
    for $C, \beta = O(1)$. Then, by our assumption in \cref{eqn:synth-assumption}, the number of copies of $\ket{B}$ that could be synthesized with high probability per copy of $\ket{\psi}$ is 
    \begin{equation}
       \Omega\left(\frac{\log^{(1+c)}(s \beta n/\alpha)}{\mathcal{M}(\ket{B}\bra{B})}\right).
    \end{equation}
    On the other hand, if $\ket{\psi_{f,S}}$ were a subset phase state with subset size $\abs{S} = 2^{\log^{1+c/2}(n)}$, combining \cref{thm:robustness} with \cref{eq:s-mbound} allows us to conclude
    \begin{equation}
        \mathcal{M}(\psi_{f,S}) \leq s'\log^{1+c/2}(n),
    \end{equation}
    then the upper bound in \cref{eq:UB} translates into being able to synthesize no more than
    \begin{equation}
        O\left(\frac{\log^{(1+c/2)} (n)}{\mathcal{M}(\ket{B}\bra{B})}\right)
    \end{equation}
    copies of $\ket{B}$ per copy of $\ket{\psi}$. %
    Then, a potential distinguisher handed $\poly(n)$ copies of an unknown $\ket{\psi}$ which is either one of these two types of states could simply run the black-box state synthesis protocol, count the number of copies of $\ket{B}$ synthesized with a SWAP test with a pre-existing copy of $\ket{B}$ -- assumed to be a canonical state such as $\ket{T}$ -- and thereby distinguish between these two states with high probability. Since the subset phase states are indistinguishable from Haar random states (given polynomially many copies), this is a contradiction.
\end{proof}
The above statement is deliberately agnostic concerning the magic monotone used, and is applicable for mixed resource states $\rho$
as well. Indeed, the powerful feature of constructing pseudomagic states is that, as pure states, no-go results can be proven using pseudomagic pure states and they remain valid for any state $\rho$ with appropriate extensions of magic monotones to mixed states. Moreover, it is well-known that while stabilizer entropies serve as effective magic measures (see \smref{sec:prelim} for details) for pure states, they fall short in quantifying the distillable magic for mixed states. This discrepancy arises because the set of states which have zero stabilizer entropy does not align with the corresponding set for other magic monotones, particularly those discussed in Table 1 of the main text. Whereas the set of states for which stabilizer entropies is zero encompasses stabilizer-mixed states, as detailed in \smref{sec:prelim}, the set of zero magic monotones in Table 1 encompasses convex combinations of stabilizer states, known as the \emph{stabilizer polytope}. However, the theorem mentioned above remains applicable to any magic monotone that is only bounded by some $\alpha=\Theta(1)$ stabilizer entropy and the robustness of magic for pure states. As we discuss below, whether this bound extends to mixed states as well is unimportant.

To be more specific, let us examine a significant example. Two prominent magic monotones used for distilling mixed states and for quantifying quantum speed-ups~\cite{PRXQuantum.2.010345} are the max-relative entropy of magic denoted as $D_{\max}(\rho)$, and the mixed-state extent denoted as $\Xi(\rho)$. The former has been introduced in Table 1 of the main text, while the latter is an extension of the stabilizer extent $\xi(\psi)$ defined for pure states. This extension, as commonly practiced in quantum resource theories, is achieved through a process known as \emph{convex-roof} extension~\cite{PRXQuantum.2.010345}
\begin{equation}
    \Xi(\rho)\coloneqq  \log\inf \left\{\sum_j p_j \xi( \psi_j) \mid \rho = \sum_{j} p_j  \ket{\psi_j}\bra{\psi_j}\right\},
\end{equation}
applicable for mixed states, where the infimum is taken over pure state decompositions of $\rho$, as the convex hull or the convex roof of the stabilizer extent $\xi$, defined in Table 1.
For pure states $\psi$, it holds that $\Xi(\psi)=D_{\max}(\psi)=\xi(\psi)$, i.e., both the max-relative entropy and the mixed-state extent reduce to the stabilizer extent $\xi(\psi)$ .

This enables us to further explore the consequences of \cref{MSD}. While for pure states $\psi$ it holds that
\begin{equation}
    \Xi(\psi)=D_{\max}(\psi)=\xi(\psi)\le \mathcal{R}(\psi),
\end{equation}
for mixed states $\rho$, the bounds are completely different~\cite{PRXQuantum.2.010345}:
\begin{equation}
    D_{\max}(\rho)\le \Xi(\rho),\quad \log\left[2^{D_{\max}(\rho)+1}-1\right]\le \mathcal{R}(\rho)
\end{equation}
whereas no known bound exists between $\Xi(\rho)$ and $\mathcal{R}(\rho)$. Similarly, no known bound exists between $\Xi$ and $M_{\alpha}$, as well as between $D_{\max}$ and $M_{\alpha}$ for mixed states. However, by virtue of \cref{MSD}, the following corollary holds.

\begin{corollaryS}[Mixed state
magic] The mixed-state extent
$\Xi$ and the max-relative entropy $D_{\max}(\rho)$ satisfy the conditions of
\cref{MSD} for $\alpha=2$ and thus they can be used as a
measure to bound magic state distillation.
\end{corollaryS}
The above corollary reflects what was discussed earlier, i.e., the strength of our pseudomagic construction lies in the ability to leverage established bounds applicable to pure states, thereby drawing inferences for magic-state distillation applied to mixed states. Perhaps the most remarkable example of this is the observation that, despite stabilizer entropies having distinct null sets in the context of mixed states compared to other magic monotones, they can still be employed to draw conclusions regarding general magic-state distillation protocols.

It is also worth mentioning that in practical magic state distillation, one usually does not require perfect magic states $\ket{B}\bra{B}$ to be distilled,
but this is expected to be possible only approximately with high probability~\cite{PRXQuantum.2.010345,PhysRevLett.125.060405,Howard2017}.
This is reminiscent of approximate asymptotic
state transformations in entanglement theory. We here point out that statements of the above form can be made robust in the
appropriate sense to capture such approximate state transformations.
In the situation where $\rho^{\otimes k}$ is not exactly distilled to $\ket{B}\bra{B}^{\otimes m}$,
but where this happens probabilistically with some probability $p\in [0,1]$ in an approximate fashion, so that there exists an $\epsilon>0$
such that $\rho^{\otimes k}$ is transformed into $p\tau$ with the property that
\begin{equation}
    \tr(\tau \ketbra{B}{B}^{\otimes m})\geq 1-\epsilon,
\end{equation}
then it necessarily holds that (see Ref.~\cite{PRXQuantum.2.010345})
\begin{equation}
    k\geq \frac{\log (p)+\log(1-\epsilon)+m \log F(\ket{B}\bra{B})^{-1}}{D_{\max}(\rho)},
\end{equation}
with $F(\ket{B}\bra{B})$ being the stabilizer fidelity of $\ket{B}\bra{B}$. This can be used to state a robust version of the above
theorem, in the sense of allowing for approximate
state transformations.

\section{Implications to nonstabilizer testing}\label{sec:th4}

A property test is an algorithm that checks whether a given unknown state has a particular property or not. We now focus on property testers for non-stabilizerness and present a theorem that characterizes the requirements for any such tester.

\begin{definitionS}[Non-stabilizer tester] An algorithm $\mathcal{A}_{\mathcal{M}}$ is a non-stabilizer tester if given $K$ copies of a state $\ket{\psi}$ and two threshold values $m,M\in[0,n]$ with $m<M$, obeys the following conditions.
\begin{itemize}
    \item If $\mathcal{M}({\psi})\ge M$, then $\mathcal{A}_{\mathcal{M}}$ accepts $\ket{\psi}$ with probability $\ge\frac{2}{3}$.
    \item If $\mathcal{M}({\psi})< m$, then $\mathcal{A}_{\mathcal{M}}$ accepts $\ket{\psi}$ with probability $\le\frac{1}{3}$.
\end{itemize}
The conditions above are referred to \emph{completness} and \emph{soundness} conditions, respectively.
\end{definitionS}
In this section, we prove the following
theorem.

\begin{theoremS}[Copies required for non-stabilizer testers]\label{thm:testing-k}
    Any tester for the stabilizer entropy $M_{\alpha}$ requires $K=\Omega(2^{\frac{m}{4+c}})$ (for $c>0$) copies for any $m<M<\frac{n}{4\alpha}$.
    \begin{proof}
        First of all, by plugging $\beta=1/4$ into \cref{lem:levy}, we see that for a Haar random state vector $\ket{\psi}$,
        \begin{equation}
            \Pr_{\ket{\psi} \sim \haar}\left[M_{\alpha}({\psi})<\frac{n}{4\alpha}\right]<2^{-\Omega(2^{n/2})}.
        \end{equation}
        Consequently, one has that for any $M< n/(4\alpha)$ one can write $\Pr_{\ket{\psi} \sim \haar}\left[M_{\alpha}({\psi})<M\right]<2^{-\Omega(2^{n/2})}$. Let us now define the ensembles
        \begin{equation}
            \begin{aligned}
                \mathcal{E}_{M_{\alpha}\ge M}&\coloneqq \{\ket{\psi} \mid M_{\alpha}({\psi})\ge M\},\\
                \mathcal{E}_{M_{\alpha}< M}&\coloneqq \{\ket{\psi} \mid M_{\alpha}({\psi})< M\}
            \end{aligned}
        \end{equation}
        of states with $\alpha$-stabilizer entropy at least $M$ and at most $M$. We use the notation $\Phi_{\mathcal{E}}^{(K)}$ to denote the ensemble average
        \begin{equation}
            \Phi_{\mathcal{E}}^{(K)}=\E_{\ket{\psi} \sim \mathcal{E}}\left[\st{\psi}^{\otimes K}\right],
        \end{equation}
        where $\sim$ denotes the uniform distribution over $\mathcal{E}$.
        We show that the ensemble $\mathcal{E}_{M_{\alpha}>M}$ is exponentially close to the ensemble of Haar random state, i.e.,
        \begin{equation}
            \norm{\Phi_{\mathcal{E}_{M_\alpha\ge M}}^{(K)}-\Phi_{\haar}^{(K)}}_1= O(2^{-\Omega(2^{n/2})}).
        \end{equation}
        First of all note that denoting $\pi_{M}\coloneqq \Pr_{\ket{\psi} \sim \haar}\left[M_{\alpha}({\psi})<M\right]$, we have $\pi_{M}\le 2^{-\Omega(2^{n/2})}$. Notice that we can express $\Phi_{\haar}^{(K)}$ for every $K$ as
        \begin{equation}
            \Phi_{\haar}^{(K)}=\pi_{M}\Phi_{\mathcal{E}_{M_\alpha< M}}^{(K)}+(1-\pi_{M})\Phi_{\mathcal{E}_{M_\alpha\ge M}}^{(K)},
        \end{equation}
        thus the following chain of equalities and inequality holds,
        \begin{equation}
            \begin{aligned}
                \norm{\Phi_{\mathcal{E}_{M_\alpha\ge M}}^{(K)}-\Phi_{\haar}^{(K)}}_1 &= \norm{\Phi_{\mathcal{E}_{M_\alpha\ge M}}^{(K)}-\pi_{M}\Phi_{\mathcal{E}_{M_\alpha< M}}^{(K)}-(1-\pi_{M})\Phi_{\mathcal{E}_{M_\alpha\ge M}}^{(K)}}_1\\
                &= \norm{-\pi_{M}\Phi_{\mathcal{E}_{M_\alpha< M}}^{(K)}+\pi_{M}\Phi_{\mathcal{E}_{M_\alpha\ge M}}^{(K)}}_1\\
                &= \pi_{M}\norm{\Phi_{\mathcal{E}_{M_\alpha< M}}^{(K)}-\Phi_{\mathcal{E}_{M_\alpha\ge M}}^{(K)}}_1\le \pi_{M}.
            \end{aligned}
        \end{equation}
        Now, consider the set $\mathcal{E}_{f,S}$ of subset phase states, i.e.,
        \begin{equation}
            \ket{\psi_{f,S}}=\frac{1}{\sqrt{\abs{S}}}\sum_{x\in S}(-1)^{f(x)}\ket{x},
        \end{equation}
        where both the Boolean function $x\mapsto f(x)$ and the subset $S$ have been chosen at random. Then, from Theorem $2.1$ of Ref.~\cite{aaronson_quantum_2023} we have that for every $K<\abs{S}<d$
        \begin{equation}
            \norm{\Phi_{\mathcal{E}_{f,S}}^{(K)}-\Phi_{\haar}^{(K)}}_1=O\left(\frac{K^2}{\abs{S}}\right),
        \end{equation}
        and consequently one has
        \begin{equation}
            \norm{\Phi_{\mathcal{E}_{M_\alpha\ge M}}^{(K)}-\Phi_{\mathcal{E}_{f,S}}^{(K)}}_1\le \norm{\Phi_{\mathcal{E}_{f,S}}^{(K)}-\Phi_{\haar}^{(K)}}_1+\norm{\Phi_{\mathcal{E}_{M_\alpha\ge M}}^{(K)}-\Phi_{\haar}^{(K)}}_1=O\left(\frac{K^2}{\abs{S}}\right).\label{eq:eq65}
        \end{equation}
        The above equation implies that the ensemble of phase states and the ensemble of states with stabilizer entropy $M_{\alpha}({\psi})\ge M$ (for
        $M\le n/(4\alpha)$) cannot be distinguished by a non-stabilizer tester unless $K=\Omega(\sqrt{\abs{S}})$. Indeed, following the arguments of Ref.~\cite{haug2023pseudorandom}, by applying the Helstrom bound~\cite{bae_quantum_2015}, we know for any state discrimination protocol between two ensembles $\mathcal{E},\mathcal{E}'$ the discrimination probability given $K$ copies, is upper bounded by $\frac{1}{2}+\norm{\Phi_{\mathcal{E}}^{(K)}-\Phi_{\mathcal{E}'}^{(K)}}_1$.
        Therefore, for a non-stabilizer tester $\mathcal{A}_{\mathcal{M}}$, to have a acceptance probability $\ge 2/3$ one should have
        \begin{equation}
            \norm{\Phi_{\mathcal{E}_{M_\alpha\ge M}}^{(K)}-\Phi_{\mathcal{E}_{f,S}}^{(K)}}_1\ge \frac{1}{6}.
        \end{equation}
        To conclude the proof, observe that $\ket{\psi}\in\mathcal{E}_{f,S} \implies M_{\alpha}({\psi})\le \log \abs{S}$, so $K=\Omega\left(2^{\frac{m}{4+c}}\right)$ for $c>0$. Indeed, if we assume that $K=O\left(2^{\frac{m}{4+c}}\right)$ copies were sufficient, one could take a random subset phase state $\ket{\psi}\in\mathcal{E}_{f,S}$ with $\abs{S}=2^{m/2}$ having $M_{\alpha}({\psi})\le m$ for $c>0$ and could distinguish between $\ket{\psi}\in \mathcal{E}_{f,S}$ and $\ket{\psi}\in \mathcal{E}_{M_{\alpha}\ge M}$ with $K=O\left(2^{\frac{m}{4+c}}\right)$. However, from \cref{eq:eq65} one would need 
        \begin{equation}
        K=\Omega\left(\sqrt{\abs{S}}\right)=\Omega(2^{m/4}) 
        \end{equation}
        copies to achieve this discrimination, so we have arrived at a contradiction.
    \end{proof}
\end{theoremS}

\begin{corollaryS}[Property testers for magic monotones]
    Given any magic monotone $\mathcal{M}$ such that $\mathcal{M}\ge s M_{\alpha}$ for some $\alpha,s$, a property tester for $\mathcal{M}$ requires $K=\Omega(2^{\frac{m/s}{4+c}})$.
    \begin{proof}
        Consider a magic monotone $\mathcal{M}$ such that $\mathcal{M}\ge sM_{\alpha}$. Assuming the contrapositive, i.e., that one can use $O(2^{\frac{m/s}{4+c}})$ to distinguish whether a state vector $\ket{\psi}$ has $\mathcal{M}({\psi})\le m$. However, if one could do so, one also would have distinguished $M_{\alpha}({\psi})\le m/s$ using $O(2^{\frac{m/s}{4+c}})$ copies and this is a contradiction to \cref{thm:testing-k}.
    \end{proof}
\end{corollaryS}

Thanks to \cref{thm:tune-magic} and \cref{thm:tune-entanglement}, \cref{thm:testing-k} can be strengthened to hold against non-stabilizerness testers designed for classes of input states with bounded entanglement between $\omega(\log(n))$ and $n$. We may similarly strengthen the property testing lower bounds in Section 3 of Ref.~\cite{aaronson_quantum_2023} to hold for input states with bounded magic.

\section{The independence of pseudomagic and pseudoentanglement}

In this section, we are going to show the following results:
\begin{itemize}
    \item Given two functions $f(n),g(n)\in[\omega(\log n), O(n)]$ there exists a pseudorandom ensemble with entanglement entropy $\Theta(f(n))$ over superpolynomially many cuts and magic $\Theta(g(n))$, for any magic measure $\mathcal{M}$ obeying $\Omega(M_2)\le \mathcal{M}\le O(\mathcal{R})$ (e.g., those listed in Table I of the main text).
    \item Given any function $f(n)\in[\omega(\log n),O(n)]$, there exists a pseudomagic pair with gap $\omega(\log n)$ vs. $O(n)$ and fixed entanglement $\Theta(f(n))$ over superpolynomially many cuts. 
    \item Given any function $g(n)\in[\omega(\log n),O(n)]$, there exists a pseudoentangled pair with maximal gap $\omega(\log n)$ vs. $O(n)$ and fixed magic $\Theta(g(n))$. 
\end{itemize}

We use the following notation throughout this section: given any unitary $U$ and any ensemble of states $\mathcal{E} \coloneqq  \{\ket{\psi_k}\}$, we define the ensemble induced by applying $U$ to every member of $\mathcal{E}$ as $\mathcal{E}_U \coloneqq  \{U \ket{\psi_k}\}$.

\begin{lemmaS}[Pseudorandomness is preserved under application of efficient unitaries]\label{lem:efficientu}
For any unitary $U$ with a $\poly$-depth circuit and any pseudorandom ensemble $\mathcal{E}$, the ensemble $\mathcal{E}_U$ is also a pseudorandom ensemble.
\begin{proof}
    Since $U$ is efficient by assumption, the new ensemble automatically has the two necessary properties of a pseudorandom ensemble:
    \begin{itemize}
        \item \emph{Efficient preparability}: we can prepare $\mathcal{E}_U$ by first preparing $\ket{\psi_k}$ and then applying the unitary $U$.
        \item \emph{Computational indistinguishability}: assume for contradiction's sake that there exists some efficient distinguisher $\mathcal{A}$ and some $t = O(\poly n)$ such that
        \begin{equation}
            \abs{\Pr_{\ket{\psi} \sim \mathcal{E}_U}[\mathcal{A}\big(\ket{\psi}^{\otimes t}\big)=1] - \Pr_{\ket{\phi} \sim \haar}[\mathcal{A}\big(\ket{\phi}^{\otimes t}\big)=1]} > \negl(n).
        \end{equation}
        Then, define the (efficient) distinguisher $\mathcal{A}' \coloneqq  \mathcal{A} \circ U^{\otimes t}$. Observe that $\Pr_{\ket{\psi} \sim \mathcal{E}}[\mathcal{A}'\big(\ket{\psi}^{\otimes t}\big)=1] = \Pr_{\ket{\psi} \sim \mathcal{E}_U}[\mathcal{A}\big(\ket{\psi}^{\otimes t}\big)=1]$, and similarly by the invariance of the Haar measure, $\Pr_{\ket{\phi} \sim \haar}[\mathcal{A}'\big(\ket{\phi}^{\otimes t}\big)=1] = \Pr_{\ket{\phi} \sim \haar}[\mathcal{A}\big(\ket{\phi}^{\otimes t}\big)=1]$. Therefore,
        \begin{equation}
            \abs{\Pr_{\ket{\psi} \sim \mathcal{E}}[\mathcal{D}'\big(\ket{\psi}^{\otimes t}\big)=1] - \Pr_{\ket{\phi} \sim \haar}[\mathcal{D}'\big(\ket{\phi}^{\otimes t}\big)=1]} > \negl(n),
        \end{equation}
        which contradicts the assumption that $\mathcal{E}$ is pseudorandom.
    \end{itemize}
\end{proof}
\end{lemmaS}
A special case of the above lemma is that $\mathcal{E}_U$ is pseudorandom for any Clifford unitary $U$, since all Clifford unitaries can be implemented with $O(n^2)$ gates~\cite{aaronson_improved_2004}. This allows us to prove that there exist pseudomagic pairs where the low magic ensemble has an arbitrary degree of entanglement. Let us first prove the following lemma.
\begin{lemmaS}[Pseudomagic ensembles with tunable entanglement]\label{thm:tune-entanglement}
Let $\mathcal{E}_{f,S}$ be the ensemble of subset phase state with $\log |S|=\omega(\log n)$. Given any function $f(n)\in[\Omega(\log |S|), O(n)]$, there exist a polynomial-depth unitary $U^*$ such that the ensemble $\mathcal{E}_{U^{*}}=\{U^{*}\ket{\psi_{f,S}}\,|\, \ket{\psi_{f,S}}\in\mathcal{E}_{f,S}\}$ obeys:
\begin{itemize}
    \item $\mathcal{M}(\ket{\psi})=\Theta(\log|S|)$ with high probability over $\ket{\psi}\sim\mathcal{E}_{U^*}$, where $\mathcal{M}$ is a magic measure such that $\Omega(M_2)\le \mathcal{M}\le O(\mathcal{R})$;
    \item $S(\psi_A)=\Theta(f(n))$ with probability $\ge 1-\operatorname{negl(n)}$ over $\ket{\psi}\sim\mathcal{E}_{U^*}$;
    \item $\mathcal{E}_{U^{*}}$ is pseudorandom.
\end{itemize}
\begin{proof}
    We first show that one can tune entanglement of a pseudorandom ensemble without changing magic. Therefore, one has that $\mathcal{M}(\ket{\psi})=\Theta(\log |S|)$ with high probability. We do this by acting on the states in the ensemble with a Clifford unitary, which does not change magic. For any $\Omega(\log|S| ) \leq f(n) \leq O(n)$, we will let this Clifford unitary $U$ act on the first $f(n)$ qubits in the system, and leave the rest untouched. 

    Consider the cut $A|B$, where $A$ is the first $f(n)/2$ qubits and $B$ contains the rest of the system. The 2-R\'enyi entanglement of any state $\ket{\psi}$ across this cut can be written $-\log \tr(S_A \ketbra{\psi}{\psi}^{\otimes 2})$, where $S_A$ is the swap operator on the $A$ subsystem. Now, letting $\mathcal{C}_f$ (resp. $\haar_f$) be the ensemble of Clifford (resp. Haar) unitaries over $f(n)$ qubits, for any state $\ket{\psi}$:
    \begin{equation}
        \mathbb{E}_{U \sim \mathcal{C}_f}\tr(S_A (U \ketbra{\psi}{\psi} U^\dagger)) = \mathbb{E}_{U \sim \haar_f} \tr(U^{\otimes 2} S_A (U^\dagger)^{\otimes 2} \ketbra{\psi}{\psi}^{\otimes 2}),
    \end{equation}
    and using the identity
    \begin{equation}
        \mathbb{E}_{U \sim \haar_f} U^{\otimes 2} O (U^{\dagger})^{\otimes 2} = \frac{\tr O - 2^{-f(n)}\tr(O S)}{4^{f(n)}-1} I + \frac{\tr(O S) - 2^{-{f(n)}}\tr(O)}{4^{f(n)}-1} S,
    \end{equation}
    where $S$ is the swap operator over $f(n)$ qubits. Since $\tr(S_A) = 2^{f(n)/2} \cdot 4^{f(n)/2} = 2^{3f(n)/2}$ (recall that we are only taking the trace of $S_A$ over sites where we acted with our Haar random unitary), and $\tr(S_A S)=\tr(S_B)=2^{3f(n)/2}$, we get
\begin{equation}
    \mathbb{E}_{U \sim \haar_k} U^{\otimes 2} S_A (U^{\otimes 2})^\dagger = \frac{2^{f(n)/2}}{2^{f(n)}+1}(I+S),
\end{equation}
so $\mathbb{E}_{U \sim \haar_k} \tr(S_A (U \ketbra{\psi}{\psi} U^\dagger)^{\otimes 2}) \leq \frac{2^{f(n)/2+1}}{2^{f(n)}+1} \leq 2^{-f(n)/2 + 1}$. Now, consider $2^{f(n)/4-2}$ different cuts $(A_i,B_i)$ such that $A$ is taken to be $f(n)/2$ of the first $f(n)$ qubits (there are $\binom{f}{f/2} \sim 2^f(n)$ cuts like this). By a simple union bound, we find that there must be at least one $U^* \in \mathcal{C}_f$ such that
\begin{equation}
    \mathbb{E}_{\ket{\psi} \sim \mathcal{E}_{U^*}} \tr(S_A \ketbra{\psi}{\psi}) \leq 2^{-f(n)/4}
\end{equation}
for \emph{each} of the cuts $(A_i,B_i)$. Finally, let us show that acting with a Clifford unitary on $f(n)$ qubits can change the $0$-R\'enyi entropy of at most $f(n)$. To see this, let us specialize the discussion to a subset phase state $\ket{\psi_{f,S}}\propto \sum_{x\in S}(-1)^{f(x)}\ket{x}$ and let us apply the unitary $U^{*}$ on the first $f(n)$ qubits. It is immediate to see that we can write
\begin{equation}
 U^{*}\ket{\psi_{f,S}}\propto \sum_{x\in \tilde{S}}(-1)^{f(x)}\ket{x}   
\end{equation}
where $|\tilde{S}|\le |S|2^{f(n)}$. As such, one has that the $0$-R\'enyi entropy of entanglement $S_{0}(\psi_A)\le \log|S|+k=O(f(n))$, where we used the fact that $f(n)=\Omega(\log |S|)$. Combining these two facts and using Markov's inequality, we immediately have that 
\begin{equation}
    \Pr_{\ket{\psi} \sim \mathcal{E}_{U^*}} \left[S(\psi_A) = \Theta(f(n)) \right] \geq 1 - 2^{-f(n)/4-1}=1-\operatorname{negl}(n).
\end{equation}
Furthermore, this statement holds across  $2^{f(n)/4-2}=\omega(\poly(n))$ (which is superpolynomial) cuts, as desired. 

Finally, since Clifford unitaries can be implemented efficiently, we can exploit \cref{lem:efficientu} and claim that $\mathcal{E}_{U^*}$ is pseudorandom. This concludes the proof.

\end{proof}
\end{lemmaS}

Finally, we can construct the pseudomagic pair with gap $\Theta(\poly\log n)$ vs. $\Theta(n)$ and fixed entanglement. Consider the ensemble of subset phase states $\mathcal{E}_{f,S}$ with $|S|=\omega(\poly(n))$ and construct the ensemble $\mathcal{E}_{U^{*}}$ discussed above. We just showed that $S(\psi_A)=\Theta(k(n))$ with high probability over $\ket{\psi}\sim\mathcal{E}_{U^{*}}$ and, since $U^*$ does not change magic, it holds that $\mathcal{M}(\ket{\psi})=\Theta(\log |S|)=\Theta(\poly\log n)$ for any magic measure $\Omega(M_{2})\le \mathcal{M}\le O(\mathcal{R})$.

We now prove the complementary statement: there also exist pseudoentangled ensembles where the low entanglement ensemble has tunable magic.

\begin{lemmaS}[Pseudoentangled ensembles with tunable magic]\label{thm:tune-magic}
Let $\mathcal{E}_{f,S}$ be the pseudorandom ensemble of subset phase states with $\log \abs{S} = \Theta(\poly \log n)$. For any function $g(n)\in[\Omega(\log|S|), O(n)]$, if the pseudorandom functions $f$ are 8-wise independent, there exists a product of \emph{efficient} local unitaries $\tilde{U} = \bigotimes_{i=1}^{g(n)} \tilde{U}_i$ such that the states in the ensemble $\mathcal{E}_{\tilde{U}}$ obey
\begin{itemize}
    \item $\mathcal{M}(\ket{\psi})=\Theta(g(n))$ with  probability $1-\operatorname{negl}(n)$ over $\ket{\psi}\sim\mathcal{E}_{U^*}$, where $\mathcal{M}$ is a magic measure such that $\Omega(M_2)\le \mathcal{M}\le O(\mathcal{R})$;
    \item $S(\psi_A)=\Theta(\log|S|)$ with probability $\ge 1-\operatorname{negl(n)}$ over $\ket{\psi}\sim\mathcal{E}_{U^*}$;
    \item $\mathcal{E}_{U^{*}}$ is pseudorandom.
\end{itemize}
\end{lemmaS}
\begin{proof}
In what follows, we divide the proofs in two parts. In the first part, we show tight bounds on the magic possessed by states in the ensemble $\mathcal{E}_{U^*}$. In the second part, we prove the pseudorandomness of $\mathcal{E}_{U^*}$.

\textbf{(I)} Let $\haar^{\otimes g}$ be the distribution over a product of $g$ independent single qubit Haar random unitaries. These single-qubit Haar random unitaries will act on the first $g(n)$ qubits, with an identity acting on the rest.

We will show that
    \begin{equation}
        Z \coloneqq  \frac{1}{2^n} \E_{\ket{\psi} \sim \mathcal{E}} \E_{U \sim \haar^{\otimes g}}  \left[\sum_{P \in \mathbb{P}_n} \tr^4(P U \ketbra{\psi}{\psi} U^\dagger)\right] = O(\exp(-g(n))).
    \end{equation}
    We define $V \coloneqq \frac{1}{2^n}\mathbb{E}_{U \sim \haar^{\otimes g}}\sum_{P \in \mathbb{P}_n} \left(U^\dagger P U\right)^{\otimes 4}$, so that $Z$ can be written $Z = \mathbb{E}_{\ket{\psi} \sim \mathcal{E}} \tr(V \ketbra{\psi}{\psi}^{\otimes 4})$.
    We also define $W \coloneqq  \E_{U \sim \haar_1}\left[(U^\dagger)^{\otimes 4} (\sum_{P \in \mathbb{P}_1} P^{\otimes 4}) U^{\otimes 4}\right]$, where $\haar_1$ is the ensemble of single qubit Haar random unitaries. We can analytically calculate
    \begin{equation}
        W = \frac{8}{5} \Pi_{\textrm{sym}} + 4 \Pi',
    \end{equation}
    where $\Pi'$ is the projector onto the subspace spanned by $\ket{0101} - \ket{0110} - \ket{1001} + \ket{1010}$ and $\ket{0011} - \ket{0110} - \ket{1001} + \ket{1100}$. The details of $W$ are irrelevant: the important part is that, when written in the computational basis, all matrix elements of $W$ are bounded in magnitude by $\frac{8}{5}$. Therefore, all the matrix elements of $V$ are bounded in magnitude by $\left(\frac{4}{5}\right)^{g(n)}$.

    Using this, we rewrite
    \begin{equation}
        \begin{aligned}
            Z &= \mathbb{E}_{\ket{\psi} \sim \mathcal{E}} \tr(V \ketbra{\psi}{\psi}^{\otimes 4}) \\
            &= \frac{1}{\abs{S}^4}\sum_{\{x_j\},\{y_j\}} \E_f\left[(-1)^{\sum_j f(x_j)+f(y_j)}\right] \tr(V \ket{x_1,x_2,x_3,x_4} \bra{y_1,y_2,y_3,y_4}),
        \end{aligned}
    \end{equation}
    where now $x_j \in S$ and $y_j \in S$. Our bound on the matrix elements of $V$ tells us $\abs{\tr(V \ket{x_1,x_2,x_3,x_4} \bra{y_1,y_2,y_3,y_4})} \leq \left(\frac{4}{5}\right)^{g(n)}$, so
    \begin{equation}
        Z \leq \left(\frac{4}{5}\right)^{g(n)}\frac{1}{\abs{S}^4} \sum_{\{x_j\},\{y_j\}} \E_f\left[(-1)^{\sum_j f(x_j)+f(y_j)}\right].
    \end{equation}
    The problem just reduces to bounding the number of nonzero elements in the sum. Since $f$ is $8$-wise independent, $\E_f\left[(-1)^{\sum_j f(x_j)+f(y_j)}\right]$ is only nonzero when the $x_j, y_j$ can all be `paired up'. That is, any given bitstring must appear an even number of times in $\{x_j\} \cup \{y_j\}$. Therefore, the number of nonzero elements is at most $\abs{S}^4 \binom{8}{2,2,2,2}$, where $\binom{8}{2,2,2,2}=2520$ is a multinomial coefficient. The $\abs{S}^4$ is because we can select 4 pairs from $\abs{S}$, and there are $\binom{8}{2,2,2,2}$ ways of allocating these bitstrings amongst $x_1,x_2,x_3,x_4,y_1,y_2,y_3,y_4$. So to conclude, $Z \leq 2520\left(\frac{4}{5}\right)^{g(n)}$. For all $g(n) \geq 70$ (which we assume for the rest of the proof), $Z \leq \left(\frac{9}{10}\right)^{g(n)}$.
    To summarize, we arrive at
    \begin{equation}
        \E_{U \sim \haar^{\otimes k}}\left[\E_{\psi \sim \mathcal{E_U}}\left[\frac{1}{2^n} \sum_{P \in \mathbb{P}_n} \tr^4(P \ketbra{\psi}{\psi})\right]\right] \leq \left(\frac{9}{10}\right)^{g(n)},
    \end{equation}
    so there must be some $U^* = \bigotimes_{i=1}^n U_i$ such that $\E_{\psi \sim \mathcal{E_{U^*}}}\left[\frac{1}{2^n} \sum_{P \in \mathbb{P}_n} \tr^4(P \ketbra{\psi}{\psi})\right] \leq \left(\frac{9}{10}\right)^{g(n)}$. By Markov's inequality, we see that
    \begin{equation}
        \Pr_{\ket{\psi} \sim \mathcal{E_{U^*}}}\left[\frac{1}{2^n} \sum_{P \in \mathbb{P}_n} \tr^4(P \ketbra{\psi}{\psi}) \geq \left(\frac{19}{20}\right)^{g(n)}\right] \leq \left(\frac{18}{19}\right)^{g(n)},
    \end{equation}
    or
    \begin{equation}
        \Pr_{\ket{\psi} \sim \mathcal{E_{U^*}}}\left[M_2({\psi}) \leq \log_2(20/19) g(n)\right] \leq \left(\frac{18}{19}\right)^{g(n)},
    \end{equation}
    so the states of the pseudorandom ensemble $\mathcal{E_{U^*}}$ have second stabilizer entropy $M_2 = \Omega(k(n))$ with overwhelming probability. The action of the $g(n)$ Haar random unitaries can also only increase magic by at most $O(g(n))$. To see this, let us write the action of $U^{*}$ on $\ket{\psi_{f,S}}$ in the computational basis. It is immediate to see that
    \begin{equation}
U^{*}\ket{\psi_{f,S}}\propto\sum_{x\in\tilde{S}}(-1)^{f(x)}\ket{x}  \,,
    \end{equation}
    where $\tilde{S}\subset\{0,1\}^{n}$ with cardinality bounded $|\tilde{S}|\le |S|2^{k}$. As such, with techniques similar to the one used in \cref{sec:pseudomagicproofs}, and using the fact that $g(n)=\Omega(\log |S|)$, we can bound $\mathcal{R}(\ket{\psi})=O(g(n))$ and $M_{0}(\ket{\psi})=O(g(n))$. Therefore, for any measure $\mathcal{M}$ obeying $\Omega(M_2)\le \mathcal{M}\le O(\mathcal{R})$, we can write
\begin{equation}
     \Pr_{\ket{\psi} \sim \mathcal{E_{U^*}}}\left[\mathcal{M}({\psi}) =\Theta(g(n))\right] \geq 1-\operatorname{negl}(n)
\end{equation}
    \textbf{(II)} So far, we have established that for the ensemble $\mathcal{E}$ of pseudorandom subset phase states, there exists a unitary $U^* = \bigotimes_{i=1}^n U_i^*$ such that $\mathcal{E}_{U^*}$ has both high magic and unchanged entanglement. In order to finally show that $\mathcal{E}_{U^*}$ is also pseudorandom, however, via \cref{lem:efficientu}, we still have to show that $U^*$ can be constructed efficiently with a finite gate set. Luckily, the Solovay-Kitaev theorem guarantees that we can approximate even an inefficient $U^*$ in polynomial depth. Combining this with the continuity of pseudorandomness (\smref{sec:prs-cont}), entanglement, and magic, we can show that we do not need to exactly implement $U^*$. By the Solovay-Kitaev theorem, we can approximate each of the $U_i^*$ with a unitary $\tilde{U}_i$ such that
    \begin{equation}
        \norm{U_i^*-\tilde{U}_i} \leq \frac{2^{-3n}}{n}
    \end{equation}
    and $\tilde{U}_i$ has $O(n \log n)$ depth. Then, defining $\tilde{U} \coloneqq  \bigotimes_{i=1}^n \tilde{U}_i$, we see that
    \begin{equation}
        \norm{U^* - \tilde{U}} \leq 2^{-3n}
    \end{equation}
    and $\tilde{U}$ is still a $O(n \log n)$ depth circuit.
    \end{proof}

We can now present the main results listed at the beggining of the section. The technical part is contained in \cref{thm:tune-entanglement} \cref{thm:tune-magic}. Let us prove the following.
\begin{theoremS}
For any two functions $f(n), g(n)$ there exists a pseudorandom ensemble with entanglement entropy $\Theta(f(n))$ over superpolynomially many cuts and magic $\Theta(g(n))$, for any measure of magic $\Omega(M_2)\le \mathcal{M}\le O(\mathcal{R})$. 
\begin{proof}
We have three different cases two discuss. If $f(n)=\Theta(g(n))$, then it is sufficient to pick the pseudorandom ensemble of subset phase state $\mathcal{E}_{f,S}$ with $|S|=\Theta(\exp f(n))$. Then, if $f(n)=\omega(g(n))$, we consider the ensemble of subset phase states $\mathcal{E}_{f,S}$ with $|S|=\Theta(\exp g(n))$ and use \cref{thm:tune-entanglement} to construct ensemble $\mathcal{E}_{U}$ and boost the entanglement to $\Theta(f(n))$ with high probability over superpolynomially many cuts. Finally, if $g(n)=\omega(f(n))$, we consider the pseudorandom ensemble of subset phase states $\mathcal{E}_{f,S}$ with $|S|=\Theta(\exp f(n))$ and use \cref{thm:tune-magic} to construct an ensemble $\mathcal{E}_{U}$ to boost magic $\mathcal{M}(\ket{\psi})=\Theta(g(n))$ with high probability over $\ket{\psi}\sim \mathcal{E}_{U}$. This concludes the proof. 
\end{proof}
\end{theoremS}

\begin{theoremS}\label{th:newtun}
Given any function $f(n)\in[\omega(\log n),O(n)]$ there exists a pseudomagic pair with maximal gap $\Theta(\poly\log n)$ vs. $\Theta(n)$ and fixed entanglement $\Theta(f(n))$.
\begin{proof}
Let us first construct the high magic ensemble  first, that we define $\mathcal{E}_1$. Let us start from the pseudorandom subset phase state ensemble $\mathcal{E}_{f,S_{1}}$ with $|S_1|=\Theta(\exp f(n))$. We know, see Ref.~\cite{aaronson_quantum_2023}, that $S(\psi_A)=\Theta(f(n))$ with high probability over $\ket{\psi}\sim\mathcal{E}_{f,S_1}$  over superpolynomially many cuts. Then, let us use \cref{thm:tune-magic} that shows the existence of a product unitary $U_1$ such that the ensemble $\mathcal{E}_{1}\coloneqq \{U_1\ket{\psi_{f,S_1}}\,|\, \ket{\psi_{f,S_1}}\in\mathcal{E}_{f,S_1}\}$ have $\mathcal{M}(\ket{\psi})=\Theta(n)$ with high probability over $\ket{\psi}\sim \mathcal{E}_1$. Therefore the ensemble $\mathcal{E}_1$ is characterized by states having $\Theta(f(n))$ entanglement and $\Theta(n)$ magic. Let us now construct the low magic ensemble with entanglement $\Theta(f(n))$. Let us again start from the pseudorandom ensemble of subset phase states $\mathcal{E}_{f,S_2}$ with $S_{2}=\omega(\log n)$. Invoking \cref{thm:tune-entanglement}, applying a specific Clifford unitary $U_2$, we can construct an ensemble $\mathcal{E}_2\coloneqq \{U_2\ket{\psi_{f,S_2}}\,|\, \ket{\psi_{f,S_2}}\in\mathcal{E}_{f,S_1}\}$ featuring magic $\mathcal{M}(\ket{\psi})=\Theta(\log |S_2|)$ and entanglement $S(\psi_A)=\Theta(f(n))$ over superpolynomially many cuts and with high probability over $\ket{\psi}\sim\mathcal{E}_2$. Therefore $\mathcal{E}_2$ has low magic $\Theta(\poly\log n)$ and fixed entanglement $\Theta(f(n))$. As such, the pair $\mathcal{E}_1,\mathcal{E}_2$ is a pseudomagic pair with maximal gap $\Omega(\poly\log n)$ vs. $\Theta(n)$ and fixed entanglement $\Theta(f(n))$.
\end{proof}
\end{theoremS}
\begin{theoremS}\label{th:newtun2}
Given any function $g(n)\in[\omega(\log n),O(n)]$ there exists a pseudoentangled pair with maximal gap $\Theta(\poly\log n)$ vs. $\Theta(n)$ and fixed magic $\mathcal{M}=\Theta(g(n))$, for any measure $\Omega(M_2)\le \mathcal{M}\le O(\mathcal{R})$.
\begin{proof}
The proof is identical to the one of \cref{th:newtun}.
\end{proof}
\end{theoremS}

\begin{remark}
We stress that \cref{th:newtun}, as well as \cref{th:newtun2}, can be generalized. In particular, we can show the existence of a pseudomagic pair with arbitrary gap (not maximal) and fixed entanglement. Conversely, the existence of pseudoentangled pair with arbitrary gap and fixed magic. The proof easily follows from the one of \cref{th:newtun}.
\end{remark}

In the following two theorems, as informally anticipated in the main text, we enhance the distillation theorem bounds by leveraging the independence of pseudomagic and pseudoentanglement

\begin{theoremS}[Prior knowledge of magic doesn't help entanglement distillation]
\label{entg_distillation}
Consider an entanglement distillation protocol that distills EPR pairs from states drawn from an ensemble $\{\psi_k\}$. Even if we are guaranteed that the states $\psi_k$ have magic $\Theta(g(n))$ for $\Omega(\poly\log n)\le g(n)\le O(n)$ with overwhelming probability, the protocol can distill at most $O(\log^{1+c} S(\rho))$ Bell pairs with high probability, where $S(\rho)$ is the von Neumann entanglement entropy across a bipartition of the system that is linear in $n$.
\end{theoremS}
\begin{proof}
The proof is identical to the one presented in Ref.~\cite{aaronson_quantum_2023}, and directly follows from \cref{th:newtun2}, i.e. the existence of a pseudoentangled pair, with maximal gap $\Theta(\poly\log n)$ vs. $\Theta(n)$, yet fixed magic $\mathcal{M}=\Theta(g(n))$.
\end{proof}

An identical strengthening of the black-box magic state distillation protocol in \cref{MSD} is also possible, using a similar proof technique.
\begin{theoremS}[Prior knowledge of entanglement doesn't help magic distillation]\label{thm:high-ent-mag-distill}
Consider a magic distillation protocol that distills copies of a target state $\ket{B}$ from states drawn from an ensemble $\{\psi_k\}$. Even if we are guaranteed that, with overwhelming probability, the states $\psi_k$ have entanglement $\Theta(n)$ across superpolynomially many cuts $(A_i,B_i)$, where the size of both cuts have size $\Omega(f(n))$ for any function $f(n)\in[\omega(\log n),O(n)]$, the number of required copies of the input state is still lower bounded by the constraints given in \cref{eq:mag-distill}.
\begin{proof}
The proof is identical to the on in \cref{MSD}, and ultimaltely is implied by the existence of a pseudomagic pair with fixed entanglement $\Theta(f(n))$, see \cref{th:newtun}.
\end{proof}
\end{theoremS}

To conclude this section, let us elucidate why the results presented herein not only contribute to advancing the findings of Ref. \cite{aaronson_quantum_2023}, specifically in \cref{entg_distillation}, but are also remarkable. The bounds derived in \cref{entg_distillation} and \cref{thm:high-ent-mag-distill} are doubly exponentially tighter than those presented in Ref. \cite{aaronson_quantum_2023} and \cref{MSD}. In particular, \cref{entg_distillation} establishes that there is no agnostic protocol capable of distilling an optimal number of Bell pairs, i.e., $\Theta(S)$, but this quantity is drastically reduced to $\Theta(\poly\log S)$ in our results. Compared to the outcome of Ref.~\cite{aaronson_quantum_2023}, our bound narrows down to a very negligible fraction (slices) of states in the Hilbert space, making it significantly stronger. Employing similar calculations throughout the paper, we can compute the fraction of states with bounded magic $g(n)$, showing that is doubly-exponentially small.
\begin{lemmaS}
For any magic measure $\mathcal{M}$ obeying $\mathcal{M}\ge \Omega(M_2)$, the fraction of pure states $\ket{\psi}$ featuring magic $\mathcal{M}(\ket{\psi})=\Theta(g(n))$ is doubly exponentially small. In formuale
\begin{equation}
 \Pr\left[\mathcal{M}(\psi) = \Theta(g(n))\right] \leq \exp(-\Omega(2^{n/3}))
\end{equation}
for any $g(n)\le n/3$. 
\begin{proof}
Using the result in \cref{lem:levy}, we know that 
\begin{equation}
    \Pr_{\ket{\psi}\sim \operatorname{Haar}}(M_{2}(\ket{\psi})\le k)\le e^{-C 2^n k^{-2}}
\end{equation}
Therefore choosing $k=g(n)$, and using the hypothesis that $g(n)\le n/3$, we find that the fraction of states featuring $M_{2}=O(g(n))$ is doubly exponentially small. As such, also the class of states featuring $M_{2}=\Theta(g(n))$ is doubly exponentially small. Moreover, this holds for every $\mathcal{M}\ge \Omega(M_2)$. In formulae
\begin{equation}
    \Pr\left[\mathcal{M}(\psi) = \Theta(g(n))\right] \leq \exp(-\Omega(2^{n/3}))
\end{equation}
\end{proof}
\end{lemmaS}

Similar findings hold for entanglement. However, as the proof is analogous to the one presented above, we omit both the statement and its proof.

\end{document}